\documentclass[preprint,12pt]{elsarticle}
\usepackage{enumerate}
\usepackage{epsfig,alltt}
\usepackage[english]{babel}
\usepackage{blindtext}
\usepackage{dsfont}
\usepackage{amssymb,amsfonts,graphicx,amsmath,amsthm}
\usepackage{tikz}
\usepackage{diagbox} 
\usetikzlibrary{bayesnet}
\usepackage[colorlinks=true,linkcolor=blue,citecolor=blue,pdfborder={0 0 0}]{hyperref}
\usepackage{natbib}
\usepackage[utf8]{inputenc}
\usepackage{soul}
\everymath{\displaystyle}
\usepackage[left=2cm,right=2.5cm,top=2.5cm,bottom=2.5cm]{geometry}

\usepackage{floatflt}
\usepackage{makeidx}
\usepackage{algorithm}
\usepackage[noend]{algpseudocode}

\usepackage{url}
\usepackage{float} 

\newtheorem{prop}{\bf Proposition}
\newtheorem{lem}{\bf Lemma}

\newtheorem{defi}{\bf Definition}

\theoremstyle{remark}

\usepackage{xcolor}
\usepackage[tikz]{bclogo}
\usepackage{hyperref}
\usepackage{blindtext}
\usepackage{tabularx}
\usepackage{color}

\usepackage{orcidlink}
\newcommand{\begeq}[1]{\begin{equation} \label{#1}}
	\newcommand{\fineq}{\end{equation}}

\usepackage{amsmath,cases}
\usepackage{caption,subcaption}
\begin{document}

\newgeometry{left=2cm, right=2.5cm, top=2.5cm, bottom=2.5cm}

\begin{frontmatter}

\title{Dependent Censoring Based on Geometric Optimization}

\author[aff1]{Anis Fradi\corref{cor1}\,\orcidlink{0000-0002-4333-2318}}

\author[aff2]{Salima Helali}

\author[aff3]{Bilel Bousselmi}

\cortext[cor1]{Corresponding author.\\
  Email: \url{anis.fradi@univ-lyon2.fr} $|$
  ORCID: \url{https://orcid.org/0000-0002-4333-2318}}

\address[aff1]{Universit\'{e} Lumi\`{e}re Lyon~2, Universit\'{e} Lyon~1,
  ERIC, 69007 Lyon, France}
\address[aff2]{Universit\'{e} de Technologie de Compi\`{e}gne, LMAC,
  F-60203 Compi\`{e}gne Cedex, France}
\address[aff3]{\'{E}cole d'Ing\'{e}nieur G\'{e}n\'{e}raliste ESME,
  F-69382 Lyon, France}

\begin{abstract}
In survival analysis, dependent censoring poses significant challenges
in accurately estimating model parameters and survival functions. This
study introduces a novel framework leveraging Extended Generalized
Marshall–Olkin (EGMO) models to address dependent censoring mechanisms.
Geometric optimization techniques are employed to develop efficient
estimation procedures that capture dependencies between failure and
censoring times. We establish their asymptotic properties. Simulation
studies and real-data applications illustrate the method's robustness
and effectiveness.  

%The Python code associated with the paper is available at: \href{https://github.com/anisfradi/Dependent-Censoring-Based-on-Geometric-Optimization}{Dependent-Censoring-Based-on-Geometric-Optimization}.
\end{abstract}

\begin{keyword}
Dependent censoring \sep EGMO models \sep Survival functions \sep
Geometric optimization \sep Maximum likelihood estimation \sep
Newton–Raphson
\end{keyword}

\end{frontmatter}

\section{Introduction}
Survival analysis provides a statistical framework for modeling time‑to‑event data arising in fields such as medicine, reliability engineering, and demography, where the primary interest lies in the distribution of the time until the occurrence of a specified event (for example, death, relapse, or system failure) (\cite{Kal02}; \cite{Kle03}). In many practical applications, the exact event time is not fully observed for all individuals, leading to various forms of censoring, including right, left, and interval censoring, which must be properly accounted for in the analysis (\cite{Fle13}; \cite{Tur76}). Consequently, standard survival methods such as the Kaplan–Meier estimator and the Cox proportional hazards model are typically developed under the assumption that the censoring mechanism is independent of the underlying failure time, conditional on covariates, an assumption that may be violated in real‑world studies (\cite{Kap58}; \cite{Cox72}; \cite{Kal02}).

In practice, the censoring mechanism may be dependent on the underlying survival time or on covariates related to prognosis, in which case the censoring time carries information about the event time and the usual independent censoring assumption is violated (\cite{Gil97};  \cite{Rob94}). Such informative or dependent censoring is common in clinical studies with dropout related to deterioration of health status or treatment side effects, and in reliability settings where preventive maintenance or early withdrawal of units depends on their latent failure risk (\cite{Tsi06}; \cite{Ibr13}). When dependence between failure and censoring times is ignored, standard estimators such as the Kaplan–Meier or Cox model–based estimators can be biased, classical tests may lose validity, and identification of the target survival distribution becomes challenging (see \cite{Tsa85}). To address these issues, several authors have proposed models and estimation strategies for dependent censoring, including copula-based formulations for joint event–censoring distributions, frailty and shared random-effect models, joint models for longitudinal and survival data, and inverse-probability-of-censoring weighting approaches that aim to recover unbiased inference under covariate-dependent censoring (\cite{Rob94}; \cite{Laa03}).

To address the challenges posed by dependent censoring, flexible parametric families of lifetime distributions have been developed that can jointly model both event and censoring times while accommodating complex hazard behaviors. Extended Generalized Marshall–Olkin (EGMO) models form a flexible parametric family of lifetime distributions designed to accommodate a wide variety of hazard rate shapes, including increasing, decreasing, bathtub, and unimodal patterns, making them suitable for modeling complex failure time behaviors (\cite{Mar05}; \cite{Dia16}). These models extend the original Marshall–Olkin family by incorporating additional shape parameters through exponentiation and generalization techniques, which enhance their tail behavior and provide greater flexibility compared to classical exponential or Weibull distributions.

To exploit the flexibility of EGMO models in the presence of dependent censoring, we introduce a geometric representation that provides an elegant and unified framework for joint modeling. This representation induces a natural parametrization of the dependence mechanism under a unit constraint, arising from the monotonicity of the survival function (\cite{pmlr-v124-holbrook20a}; \cite{Fradi2024}). As a result, it enables intuitive and interpretable visualizations of the relationship between censoring and failure times through geodesic paths on the manifold (\cite{Amari-2009}; \cite{Mar00}). Although geometric representations have been previously explored in directional statistics and in copula-based constructions for multivariate survival data, their use for explicitly modeling dependent censoring within the EGMO family via a geometric framework appears to be novel.

Despite the growing interest in constrained representations for capturing dependencies in survival analysis, a significant gap remains in their integration with flexible parametric lifetime families under dependent censoring. Existing approaches typically rely on restrictive dependence structures, such as copulas or frailty models, and often lack sufficient flexibility to accommodate complex hazard shapes, and may encounter substantial inferential challenges due to identifiability issues. Moreover, they do not provide a unified geometric interpretation of the joint behavior of survival and censoring times (see \cite{Tsia75}). These limitations motivate the development of our framework, which bridges this gap by combining flexibility, interpretability, and a coherent geometric structure.

Recently, \cite{Ingrid-2024} proposed an intuitive framework for estimating the survival function through density modeling based on a squared ($L^2$)-normalized Laguerre polynomial decomposition, ensuring positivity and normalization of the resulting density. Building upon this idea, we incorporate geometric tools from information geometry to perform maximum likelihood estimation on the underlying statistical manifold rather than in a purely Euclidean parameter space. This geometric formulation naturally handles the structural constraints of the model and leads to a constrained optimization procedure adapted to the intrinsic geometry of the parameter space.

In this paper, we further derive an explicit expression of the truncated log-likelihood of the proposed model in terms of the coefficients arising from the associated decomposition. This formulation naturally leads to a constrained optimization problem on the corresponding parameter space, which we address through a geometric iterative maximum likelihood procedure. In particular, the optimization is performed under structural constraints induced by the model geometry, allowing for a stable and efficient estimation strategy.

We establish the main theoretical properties of the proposed estimator, together with the asymptotic behavior of the associated inferential procedures (see \cite{Vaart-1998}). In particular, we prove consistency and asymptotic normality under standard regularity conditions. The finite-sample performance of the estimators is investigated through extensive simulations, highlighting their robustness in small and moderate sample regimes. Finally, we demonstrate the practical relevance of the proposed methodology through applications to real-world medical data and UEFA Champions League datasets, both exhibiting complex and informative censoring mechanisms.

The remainder of the paper is organized as follows.  Section \ref{sec:model} provides the model description and notations for the EGMO-based dependent censoring framework within a geometric representation. Section \ref{sec:dep} establishes the dependent properties of the proposed model. Section \ref{sec:like} presents the likelihood decomposition. Section \ref{sec:est} develops estimation procedures and asymptotic properties. Section \ref{sec:sim} reports experimental results, including a simulation study and two real data examples.
 Section \ref{sec:con} concludes with a discussion of extensions and future work, while Section  \ref{sec:appendix} includes the proofs of the theoretical results.
%Survival analysis frequently relies on the assumption that censoring is independent, a condition often violated in practice, leading to biased parameter estimates and misleading inferences. For instance, in clinical trials, patients may withdraw from a study due to deteriorating health or adverse treatment effects, making their subsequent survival time unknown and the censoring mechanism clearly dependent on the outcome of interest. This dependency leads to a fundamental identifiability issue, first formalized by \cite{Tsia75}, who showed that dependent censoring renders marginal survival distributions nonparametrically nonidentifiable based solely on observed data. This type of dependent censoring can be naturally formulated  within the context of competing risks, where failure and censoring are viewed as competing events that terminate follow-up. Under this perspective, the Extended Generalized Marshall–Olkin (EGMO) models provide a flexible framework for characterizing the dependence between these events through latent shock variables, enabling more accurate estimation of survival function. 
\section{Model description and notations}\label{sec:model}
The EGMO model extends the Generalized Marshall–Olkin (GMO) model by capturing dependence between the individual shock arrival times of two components while maintaining the independence of the common shock's arrival time.  Formally, this means that we consider three  random variables $X_1,X_2$ and $X_3$,  by assuming dependence between 
$X_1$ and $X_2$, while maintaining the independence of $X_3$ with respect to them. The lifetimes of interest are described by
\begin{eqnarray}
\label{egmomodel}
T=\min(X_1,X_3)\quad\text{and}\quad C=\min(X_2,X_3).
\end{eqnarray}
Based on a right-censoring approach, only the smallest default time is observed. Formally, the outcomes are limited to the couple $(Y,\delta)$, defined by 
\begin{equation}\label{obs}
Y=\min(T,C)\quad\text{and}\quad\delta=\mathds{1}_{\{T\leq C\}}.
\end{equation}
For the EGMO model presented in Equation (\ref{egmomodel}), we consider the following  notations. The density function  of each $X_i$ for  $i=1, 2, 3$
is denoted by $f_i$, while its distribution function  is represented as 
$F_i$. The survival function of $X_i$, which given by $1-F_i$, is denoted by 
$\overline{F}_i$. In addition, the cumulative hazard function of $X_i$ is denoted by  $\Lambda_i$. 
For the joint behavior of 
$X_1$ and $X_2$, the joint density function is represented as  $f_{1,2}$, the joint distribution function as 
$F_{1,2}$, and the joint survival function as $\overline{F}_{1,2}$. Let $f_T$ and $f_C$ denote the  density functions of $T$ and $C$, respectively, while $F_T$ and $F_C$ represent their corresponding distribution functions. 
The distribution function of the observed random time $Y$ is denoted by $H$,  and its density function by $f_Y$.  Note that the survival function of $Y$ is given by
\begin{eqnarray}\label{survieY}
\overline{H}(t)&=& \mathbb{P}(Y>t)=\overline{F}_{1,2}(t,t) \overline{F}_3(t). 
\end{eqnarray}
Recall that the cumulative hazard function associated with any distribution function $F$ is given by
\begin{eqnarray*}
\Lambda(t)=\int_0^t\dfrac{F(du)}{\overline{F}(u^-)},\quad t\geq 0,
\end{eqnarray*}
where $\overline{F}:=1-F$ and $\overline{F}(u^-)=\lim_{s\downarrow u}\overline{F}(s)$. In contrast, the relationship between a distribution function $F$ and its cumulative hazard function $\Lambda$ can be expressed as
\begin{eqnarray*}
\label{eq::lambda}
\overline{F}(t)=:\exp\left(-\widetilde{\Lambda}(t)\right)=\exp\left(-\Lambda^c(t)+\sum_{s\leq t}\log(1-\Delta\Lambda(s))\right),    
\end{eqnarray*}
where $\Lambda^c$ and $\Delta\Lambda$ respectively denote the continuous and discontinuous parts of $\Lambda$. %\citep[see p.898 in][]{Sho86}. 
Therefore, the survival functions of $T$ and $C$ are expressed as: \begin{eqnarray*}\label{marginal}
&&\overline{F}_T(t)=:\exp(-\widetilde{\Lambda}_T(t))\quad\text{and}\quad\overline{F}_C(t)=:\exp(-\widetilde{\Lambda}_C(t)),
\end{eqnarray*}
where $\Lambda_T$ and $\Lambda_C$ define the cumulative hazard functions of $T$ and $C$ respectively.
\section{Dependent properties}\label{sec:dep}
The survival distributions of $T$ and $C$ are given by 
$$\overline{F}_T(t)=\overline{F}_1(t) \overline{F}_3(t) \quad \text{and} \quad \overline{F}_C(t)=\overline{F}_2(t) \overline{F}_3(t).$$
Based on \cite{Hel25}, the survival joint distribution  of $(T,C)$  is given by

\begin{eqnarray*}
\widetilde{P}(t,s)
=\mathbb{P}(T>t,C>s)&=&\mathbb{P}(X_1>t, X_2>s) \mathbb{P}(X_3> \max(t,s))
\\&=&
\min \left(\overline{F}_C(s)\overline{F}_T(t)^{1-\alpha(t)}, \overline{F}_T(t) \overline{F}_C(s)^{1-\beta(s)} \right) 
\Omega \left(\overline{F}_{1}(t),\overline{F}_{2}(s) \right),
\end{eqnarray*}
where $\displaystyle{\alpha=\widetilde{\Lambda}_{3}/(\widetilde{\Lambda}_1+\widetilde{\Lambda}_{3})}$, $\displaystyle{\beta=\widetilde{\Lambda}_{3}/(\widetilde{\Lambda}_2+\widetilde{\Lambda}_{3})}$ 
 and $\displaystyle{\Omega \left(\overline{F}_{1}(t),\overline{F}_{2}(s) \right)=\overline{F}_{1,2}(t,s)/(\overline{F}_1(t)\overline{F}_2(s))}$ is the Sibuya dependence function of $(X_1,X_2)$. 
In particular, the survival copula is given by
$$\widetilde{C}(u,v)=uv\min\left(u^{-\alpha\left(\overline{F}^{-1}_T(u) \right)},v^{-\beta\left(\overline{F}^{-1}_C(v) \right)}\right) \Omega \left(\overline{F}_{1}\left(\overline{F}^{-1}_T(u) \right),\overline{F}_{2}\left(\overline{F}^{-1}_C(v) \right) \right),$$
for $(u,v)\in \overline{F}_T(\mathbb{R}_+)\times \overline{F}_C(\mathbb{R}_+)$, where $\overline{F}^{-1}$ stands for the generalized inverse function of $\overline{F}$.\\
\ \\
The sub-densities $f_{Y,\delta=0}$ and $f_{Y,\delta=1}$ of the vector $(Y,\delta)$ are given by
\begin{eqnarray*}
   f_{Y,\delta=0}(t)  = \overline{F}_3(t) \int_t^{+\infty} f_{1,2}(u,t)du,
\end{eqnarray*}
and 
\begin{eqnarray*}
 f_{Y,\delta=1}(t)&=&f_Y(t)-f_{Y,\delta=0}(t)
 \\&=& \overline{F}_3(t) \int_t^{+\infty}f_{1,2}(t,y)dy +f_3(t) \overline{F}_{1,2}(t,t).
\end{eqnarray*}
To address the issue of non-parametric likelihood identifiability, we study the parametric truncated log-likelihood identifiability  based on a geometric representation with positive coefficients, as defined in Equation (\ref{estloglikelihood}). Following \cite{Hel25}, the use of the  polynomial expansions transforms the infinite-dimensional problem into a finite-dimensional one, where the parameters lie in a simplex and preserve key probabilistic properties. This approach restores identifiability asymptotically, as the polynomial degrees increase, and provides a tractable framework for maximum likelihood estimation. 
\section{Likelihood decomposition}
\label{sec:like}
In this section, we derive decompositions for both the marginal survival function and the joint density function, which in turn yield the truncated log-likelihood.
\subsection{Marginal survival function decomposition}
We consider probability distributions defined over $I=(a,b) \subseteq \mathbb{R}$ according to the Lebesgue measure.
A survival function (SF) denoted $\overline{F}$ associated to a random variable $X$ is a decreasing and differentiable function defined from $I$ into $[0,1]$. The space of survival functions,  defined on $I$, satisfies  
\begin{equation*}
\Gamma	= \Big\{\overline{F}:I \to [0,1] \hspace{0.1cm} \big|\hspace{0.1cm} \dot{\overline{F}} \leq 0,  \; \text{and}  \hspace{0.1cm} \big(\overline{F}(a),  \overline{F}(b)\big)=(1,0) \Big\}. 
\end{equation*}
Since $\Gamma$ does not admit a natural structure due to the inequality constraint, there is no well-defined metric on it. 
In connection with $\Gamma,$
the space of square-root density functions (SRDFs) is a Riemannian representation, satisfying
\begin{eqnarray}
	\label{eq:Hilbert_sphere}
	\mathcal{Q}=\Big\{q : I  \to \mathbb{R} \hspace{0.1cm} \big| \hspace{0.1cm} q \equiv \sqrt{-\dot{\overline{F}}}\geq 0, \hspace{0.1cm}   \hspace{0.1cm}||q||^2_{\mathbb{L}^2} = \int_{I} q(t)^2 dt =1 \Big\}, 
\end{eqnarray} 
where the tangent space of $\mathcal{Q}$ at an arbitrary $q$ is 
\begin{eqnarray*}
	\mathcal{T}_{q} \mathcal{Q} =\Big\{ g:I \to \mathbb{R} \hspace{0.1cm} \big | \hspace{0.1cm} \int_{I} q(t) g(t) dt=0  \Big\}.
\end{eqnarray*}
Since $q$ has a unit norm, the set $\mathcal{Q}$ forms a sub-manifold with the $\mathbb{L}^2$ metric. 
Given two tangent vectors $g_1,g_2 \in \mathcal{T}_q \mathcal{Q} $, 
the Fisher-Rao metric becomes the $\mathbb{L}^2$  metric satisfying
\begin{eqnarray*}
\big<g_1,g_2\big>_{\mathbb{L}^2}=\int_{I} g_1(t) g_2(t) dt.
\end{eqnarray*}
Note that any SRDF $q \in \mathcal{Q}$ can be (isometrically) represented by a unique SF $\overline{F} \in 	\Gamma$ expressed as
\begin{eqnarray}
	\label{eq:SF}
	\overline{F}(x)=  \int_{x}^{b} q(t)^2 dt , \hspace{0.1cm} x \in I.
\end{eqnarray}
In order to reduce the complexity of the group of survival functions due to its hard structure, we make use of  the convergent orthogonal series
expansion
based on writing the corresponding SRDF $q\in  \mathcal{Q}$ as a  linear combination of orthogonal basis $(\phi_l)_l$ in $\mathbb{L}^2(I,\rho)$ with a positive weight function $\rho(t)$
\begin{eqnarray}
	q(t)= \sum_{l=1}^{\infty}  w_l \phi_l(t),
	\label{eq:K-L}
\end{eqnarray}
We consider a truncated version of $q$ at an arbitrary order $m$, expressed as 
\begin{eqnarray}
	q_W^d(t) = \sum_{l=1}^{d} w_l \phi_l(t),
	\label{eq:truncated_K-L}
\end{eqnarray}
with the notation $W=(w_1,\dots,w_d)^T$.
Consequently, this expansion makes it more easy to check that $||q_W^d||_{\mathbb{L}^2}=1$. It translates directly to a unit constraint on the coefficients $w_l$, $l=1,\dots,d$. 
\begin{lem}\label{lemI}
The truncated version of $q$ denoted $q_W^d$ belongs to $\mathcal{Q}$ if and only if $\sum_{l=1}^{d} w_l^2 =1$. \\
\qed
\end{lem}
\noindent
 Following (\ref{eq:SF}), the truncated version of $\overline{F}$ satisfies
\begin{eqnarray}\label{appF3}
	\overline{F}_W^d (x) 
	= 
	W^T \Big( \int_{x}^{b} \Phi(t) \Phi(t)^T \rho(t) dt  \Big) W,  
\end{eqnarray}
where
 $
\Phi(t)=(\phi_1(t), \dots, \phi_d(t))^T$ $\forall t \in I$.
\subsection{Joint density function decomposition}
In this section, we focus on probability distributions defined over a bivariate interval $\mathcal{D}=(a,b)^2 \subseteq \mathbb{R}^2$ with respect to the Lebesgue measure. A probability density function (PDF) $f$, associated with a random vector $(X, Y)$ defined on $\mathcal{D}$, belongs to the space
\begin{eqnarray*}
	\mathcal{P} = \Big\{f : \mathcal{D} \to \mathbb{R} \ \big| \ f \in \mathbb{L}^1(\mathcal{D},\mathbb{R}), \ f \geq 0, \ \text{and} \ ||f||_{\mathbb{L}^1} = \int_{\mathcal{D}} f(t_1,t_2) dt_1 dt_2 = 1 \Big\}.
\end{eqnarray*}
We work within the space $\mathcal{P}$ without assuming a parametric model, instead endowing $\mathcal{P}$ with the structure of an infinite-dimensional (formal) Riemannian manifold. First, we treat it as a smooth manifold. For a given $f \in \mathcal{P}$, the tangent space can be expressed as
\begin{eqnarray*}
	\mathcal{T}_f \mathcal{P} = \Big\{ h \in C^{\infty}(\mathcal{D}) \ \big| \ \int_{\mathcal{D}} h(t_1,t_2) dt_1 dt_2 = 0 \Big\}.
\end{eqnarray*}
Once the smooth manifold and its associated tangent space are established, we define a Riemannian metric, i.e., a smoothly varying, symmetric, non-degenerate, bilinear function $\mathcal{G}(., .)_f : \mathcal{T}_f \mathcal{P} \times \mathcal{T}_f \mathcal{P} \to \mathbb{R}^+$. Riemannian metrics are essential for defining a notion of distance on a manifold that is independent of any embedding in Euclidean space.
For a given $\mathcal{D}$, the nonparametric Fisher metric on $\mathcal{P}(\mathcal{D})$ is defined as
\begin{eqnarray*}
	\mathcal{G}(h_1, h_2)_f = \int_{\mathcal{D}} \frac{h_1(t_1,t_2) h_2(t_1,t_2)}{f(t_1,t_2)} dt_1 dt_2.
\end{eqnarray*}
The manifold $\mathcal{P}$, equipped with the Fisher metric, can be challenging to work with when calculating geometric quantities of interest (e.g., geodesics or distances). To simplify these computations, we shift focus to the space of square-root density functions
\begin{eqnarray*}
	\mathcal{H} = \Big\{\psi : \mathcal{D} \to \mathbb{R} \ \big| \ \psi \equiv \sqrt{f} \geq 0, \ \text{and} \ ||\psi||_{\mathbb{L}^2}^2 = \int_{\mathcal{D}} \psi(t_1,t_2)^2 dt_1 dt_2 = 1 \Big\}.
\end{eqnarray*}
This space, which is linked to $\mathcal{P}$ via a simple transformation described below, provides a more tractable framework for calculations. By differentiating the unit condition on $\mathcal{H}$, the tangent space of $\mathcal{H}$ at any $\psi$ is given by
\begin{eqnarray*}
	T_{\psi} \mathcal{H} = \Big\{ g : \mathcal{D} \to \mathbb{R} \ \big| \ \big<\psi, g\big>_{\mathbb{L}^2} = \int_{\mathcal{D}} \psi(t_1,t_2) g(t_1,t_2) dt_1 dt_2 = 0 \Big\}.
\end{eqnarray*}
\begin{lem}
\label{lemII}
The map $L : \big(\mathcal{P}, \mathcal{G}(., .)_f\big) \to \big(\mathcal{H}, \big<.,.\big>_{\mathbb{L}^2}\big)$ defined by $L(f) = 2\sqrt{f}$ is a Riemannian isometry. \\ \qed
\end{lem}
\noindent
Let $(\xi_{l_1,l_2})_{l_1,l_2}$ be a set of orthonormal basis functions in $\mathbb{L}^2(\mathcal{D}, \rho)$. The convergent orthogonal series expansion of $\psi$ is expressed as
\begin{eqnarray*}
	\psi(t_1,t_2) = \sum_{l_1,l_2=1}^{\infty} v_{l_1,l_2} \xi_{l_1,l_2}(t_1,t_2),
\end{eqnarray*}
where $v_{l_1,l_2}$ are coefficients given by $v_{l_1,l_2} = \big<\psi, \xi_{l_1,l_2}\big>_{\mathbb{L}^2}$. A truncated version of $\psi$ at order $(m,p)$ is written as
\begin{eqnarray*}
	\psi_{V}^{m,p}(t_1,t_2) = \sum_{l_1,l_2=1}^{m,p} v_{l_1,l_2} \xi_{l_1,l_2}(t_1,t_2),
\end{eqnarray*}
with $V = \big(v_{1,1}, v_{1,2}, \dots, v_{m,p}\big)^{T} \in \mathbb{R}^{mp}$. The remainder of the series $\sum_{l_1,l_2=m+1,p+1}^\infty$ represents the approximation error.

\begin{lem}\label{lrmIII}
	The truncated version $\psi_{V}^{m,p}$ belongs to $\mathcal{H}$ if and only if $\sum_{l_1,l_2=1}^{m,p} v_{l_1,l_2}^2 = 1$. \\
    \qed
\end{lem}
\noindent The truncated version of $f$ is then given by
\begin{eqnarray}\label{appf12}
	f_V^{m,p}(t_1,t_2) = \psi_{V}^{m,p}(t_1,t_2)^2 = V^T \boldsymbol{\xi}(t_1,t_2) \boldsymbol{\xi}(t_1,t_2)^T V,
\end{eqnarray}
where $\boldsymbol{\xi}(t_1,t_2) = \big(\xi_{1,1}(t_1,t_2), \xi_{1,2}(t_1,t_2), \dots, \xi_{m,p}(t_1,t_2)\big)^T$, for all $(t_1,t_2) \in \mathcal{D}$.
\subsection{Truncated log-likelihood}
For a sample $S={(Y_i, \delta_i), i=1,\ldots,n}$, the log-likelihood is given by
\begin{eqnarray}\label{likelihood}
\mathcal{L}(f_{1,2},f_3;S) \nonumber&=& \sum_{i=1}^n \delta_i \log\{f_{Y, \delta=1}(y_i)\}+ (1-\delta_i) \log\{f_{Y, \delta=0}(y_i)\}
\\&=& \sum_{i=1}^n \delta_i \log\{f_{Y}(y_i)-f_{Y, \delta=0}(y_i)\}+ (1-\delta_i) \log\{f_{Y, \delta=0}(y_i)\},
\end{eqnarray}
where 
\begin{eqnarray}\label{lambdafydelta=0}
f_{Y,\delta=0}(t)  &=& \overline{F}_3(t) \int_t^{+\infty} f_{1,2}(u,t)du.
\end{eqnarray}
The direct maximization of this log-likelihood can be challenging, since it depends on an unknown $\overline{F}_3$ and $f_{1,2}$. Therefore, we must estimate both of these distributions in order to estimate the survival functions  $\overline{F}_T$, $\overline{F}_C$ and the joint survival function $\widetilde{P}$.\\
\ \\
In our context, and without loss of generality, we assume that the random times $X_i$ for $i=1,2,3$ have equal supports $S_{X_i}=I=(a,b)$.
If $X_3$ has continuous survival function $\overline{F}_3$, based on Equation (\ref{appF3}), the truncated version of $\overline{F}_3$ satisfies
\begin{eqnarray}\label{estF3}
	\overline{F}_{3,W}^d (x) 
	= 
	W^T \Big( \int_{x}^{b} \Phi(t) \Phi(t)^T \rho(t) dt  \Big) W,  
\end{eqnarray}
where  
$\Phi(t)=(\phi_1(t), \dots, \phi_d(t))^T$ $\forall t \in I$, is a  linear combination of orthogonal basis in $\mathbb{L}^2(I,\rho)$ with fixed coefficients $W=(w_1,\dots,w_d)^T$, such that $\displaystyle{\sum_{l=1}^d w_l^2=1}$. 
Similarly, if $(X_1,X_2)$ has a continuous
density $f_{1,2}$, based on Equation (\ref{appf12}), the truncated version of $f_{1,2}$ satisfies
\begin{eqnarray}\label{estf12}
	f^{m,p}_{1,2,V}(t_1,t_2) = \psi^{m,p}_{V}(t_1,t_2)^2 = \big( \sum_{i=1}^m \sum_{j=1}^p v_{i,j} \xi_{i,j}(t_1,t_2) \big)^2,
\end{eqnarray}
where $\boldsymbol{\xi}(t_1,t_2) = \big(\xi_{1,1}(t_1,t_2), \xi_{1,2}(t_1,t_2), \dots, \xi_{m,p}(t_1,t_2)\big)^T$, for all $(t_1,t_2) \in \mathcal{D}$, set of orthonormal basis functions in $\mathbb{L}^2(\mathcal{D}, \rho)$, with coefficients $V = \big(v_{1,1}, v_{1,2}, \dots, v_{m,p}\big)^{T} \in \mathbb{R}^{mp}$, such that $\sum_{l_1,l_2=1}^{m,p} v_{l_1,l_2}^2 = 1$.  As a result, based on Equation (\ref{estF3}) and (\ref{estf12}), the truncated version of the  sub-density $f_{Y,\delta=0}(t)$, is given by 
\begin{eqnarray}\label{estfydelta=0}
\widetilde{f}_{Y,\delta=0}^{m,p,d}(t)  &=&  \overline{F}_{3,W}^d(t) \int_t^{b} f_{1,2,V}^{m,p}(u,t)du.
\end{eqnarray}
The density $f_Y$ of the observed random time $Y$ can be estimated by the non parametric kernel method, defined by  
\begin{eqnarray}\label{kernel}
\widetilde{f}_{Y,n}(t)&=&\frac{1}{nh}\sum_{i=1}^n K\left( \frac{t-Y_i}{h} \right),    
\end{eqnarray}
where  $K$ is the kernel and $h > 0$ is  the bandwidth. 
Hence, by substituting the functions (\ref{estfydelta=0}) and (\ref{kernel}) into Equation (\ref{likelihood}), we obtain  the truncated log-likelihood, given by
\begin{eqnarray}\label{estloglikelihood}
\widetilde{\mathcal{L}}_{m,p,d}(\theta;S)&=& \frac{1}{n}\sum_{i=1}^n \delta_i \log\{\widetilde{f}_{Y,n}(y_i)-\widetilde{f}_{Y, \delta=0}^{m,p,d}(y_i)\}+ (1-\delta_i) \log\{\widetilde{f}_{Y, \delta=0}^{m,p,d}(y_i)\},
\end{eqnarray}
where $\theta=(W,V) \in (\Theta_1,\Theta_2)$, which form a vector of size $mp+d$. 
\section{Estimation and asymptotic properties }
\label{sec:est}
In this section, we present a geometric optimization procedure for computing the maximum likelihood estimator as well as its asymptotic properties.
\subsection{The maximum likelihood estimation}
\noindent Let $ A(t) $ denote a geodesic path on the unit sphere~\cite{pmlr-v124-holbrook20a} defined as
$$ \mathcal{S}^{M-1} = \Big\{A = (a_1,\dots,a_M)^\top \in \mathbb{R}^M \hspace{0.1cm} \big| \hspace{0.1cm}  \sum_{l=1}^{M} a_l^2 =1 \Big\} \subset \mathbb{R}^M, $$
starting from the point $ A_0 \in \mathcal{S}^{M-1} $ with initial velocity $ \dot{A}_0 \in T_{A_0}\mathcal{S}^{M-1} $, the tangent space of $\mathcal{S}^{M-1}$ at $ A_0 $. Its explicit expression is provided in the Appendix.
Let $ F : \mathcal{S}^{M-1} \to \mathbb{R} $ be a smooth scalar-valued function.
\begin{defi}
	Given an initial condition $ A_0 \in \mathcal{S}^{M-1} $ and a tangent vector $  \dot{A}_0 \in T_{A_0}\mathcal{S}^{M-1} $, the Riemannian Hessian of $ F $ along $  \dot{A}_0 $ is defined as
	\begin{equation*}
		\text{Hess}F ( \dot{A}_0,  \dot{A}_0) = \left. \frac{d^2}{dt^2} F(A(t)) \right|_{t=0},
	\end{equation*}
	where $ A(t) $ is the geodesic satisfying $ A(0) = A_0 $ and $ \dot{A}(0) =  \dot{A}_0 $.
\end{defi}
\begin{lem}
\label{lemIV}
	Under the above assumptions, the Riemannian Hessian of $ F $ evaluated along $  \dot{A}_0 $ is given by
	\begin{equation*}
		\text{Hess}F = F_{AA} - F_A^\top A_0 I_M,
	\end{equation*}
	where $ F_A $ and $ F_{AA} $ denote the Euclidean gradient and Hessian of $ F $, respectively.
	\qed 
\end{lem}

\noindent Algorithm~\ref{alg:snr} shows the steps for implementing the Newton-Raphson method on the sphere in order to find a minimum of $F$, henceforth called \emph{Spherical Newton-Raphson (SNR)}.
	
\begin{algorithm} 
	\caption{Newton-Raphson method on the sphere.}\label{alg:snr}
	\begin{algorithmic}[1]
		\State Initialize $ A_0 $ and $  \dot{A}_0 $
		\While{not converged}
		\State Compute $ F_A $ and $ F_{AA} $.
		\State Evaluate Hessian: $ \text{Hess}F = F_{AA} - F_A^\top A_0 I_M $.
		\State Update the velocity using the Newton step:
		$$  \dot{A} \leftarrow  \dot{A} - \epsilon \left( I_M - A A^\top \right) \text{Hess}F^{-1} \left( I_M - A A^\top \right) F_A .$$
		\State Update the position using the geodesic path:
		\State $$ A \leftarrow A \cos\left( \epsilon || \dot{A}||_2 \right) + \frac{ \dot{A}}{|| \dot{A}||_2} \sin\left( \epsilon || \dot{A}||_2 \right) .$$
		\EndWhile
	\end{algorithmic}
\end{algorithm}

\noindent Given the set of observations $S={(Y_i, \delta_i), i=1,\ldots,n}$, the log-likelihood distribution evaluated at $\theta=(W,V)$ is given by Equation (\ref{estloglikelihood}),
that we want to maximize. Let
$
A_i = \widetilde{f}_{Y,n}(y_i)$, $
B_i = \widetilde{f}_{Y,\delta=0}^{m,p,d}(y_i) = \overline{F}_{3,W}^d(y_i) \int_{y_i}^{b} f_{1,2,V}^{m,p}(u,y_i)du$ and $
M(x) = \int_{x}^{b} \Phi(t) \Phi(t)^T \rho(t) dt$.
According to these new notations
the log-likelihood becomes
\begin{equation*}
\widetilde{\mathcal{L}}_{m,p,d}(\theta;S) = \frac{1}{n}\sum_{i=1}^n
\Big(
\delta_i \log\{A_i-B_i\} - (1-\delta_i) \log\{B_i\}
\Big).
\end{equation*}
The first partial derivative of the log-likelihood w.r.t. $W$ is given by
\begin{equation*}
	\frac{\partial \widetilde{\mathcal{L}}(\theta;S)}{\partial W}
	=
-	\frac{1}{n} \sum_{i=1}^n 
	\left[
	 \delta_i \times \frac{1}{A_i - B_i} \times \frac{\partial B_i}{\partial W}
	+ (1 - \delta_i) \times \frac{1}{B_i} \times \frac{\partial B_i}{\partial W}
	\right].
\end{equation*}
Or 
\begin{eqnarray*}
	\frac{\partial B_i}{\partial W}
=	\left( \int_{y_i}^{b} f^{m,p}_{1,2,V}(u, y_i) \, du \right)
	\times
	\frac{\partial \overline{F}_{3,W}^d(y_i)}{\partial W} =
	2 \left( \int_{y_i}^{b} f^{m,p}_{1,2,V}(u, y_i) \, du \right)
	\times M(y_i) \times W.
\end{eqnarray*}
So, 
\begin{eqnarray}
	\frac{\partial \widetilde{\mathcal{L}}(\theta;S)}{\partial W}
	= - \frac{1}{n} \sum_{i=1}^n \left[
	\left( \frac{1 - \delta_i}{B_i} + \frac{\delta_i}{A_i - B_i} \right)
	\times 2 \left( \int_{y_i}^b f^{m,p}_{1,2,V}(u, y_i) \, du \right)
	\times M(y_i) \times W
	\right].
	\label{gradW}
\end{eqnarray}
The first partial derivative of the log-likelihood w.r.t. $V$ is given by
\begin{equation*}
	\frac{\partial \widetilde{\mathcal{L}}(\theta;S)}{\partial V}
	= 
	- \frac{1}{n} \sum_{i=1}^n 
	\left[
	 \delta_i \times \frac{1}{A_i - B_i} \times \frac{\partial B_i}{\partial V}
	+ (1 - \delta_i) \times \frac{1}{B_i} \times \frac{\partial B_i}{\partial V}
	\right].
\end{equation*}
Or
\begin{equation*}
	\frac{\partial}{\partial V} \left( \int_{y_i}^{b} f^{m,p}_{1,2,V}(u, y_i) \, du \right)
	=
	2 \int_{y_i}^{b} \psi^{m,p}_V(u, y_i) \times \boldsymbol{\xi}(u, y_i) \, du
\end{equation*}
implying
\begin{equation}\label{partialBV}
	\frac{\partial B_i}{\partial V}
	=
	\overline{F}_{3,W}^d(y_i)
	\times 2 \int_{y_i}^{b} \left( V^T \boldsymbol{\xi}(u, y_i) \right) \times
	\boldsymbol{\xi}(u, y_i) \, du.
\end{equation}
So,
\begin{eqnarray}
	\frac{\partial \widetilde{\mathcal{L}}(\theta;S)}{\partial V}
	= - \frac{1}{n} \sum_{i=1}^n \left[
	\left( \frac{1 - \delta_i}{B_i} + \frac{\delta_i}{A_i - B_i} \right)
	\times 2 \, \overline{F}_{3,W}^d(y_i)
	\times \int_{y_i}^b \left( V^T \boldsymbol{\xi}(u, y_i) \right)
	\boldsymbol{\xi}(u, y_i) \, du
	\right].
	\label{gradV}
\end{eqnarray}
The second partial derivative w.r.t. $W$ is
\begin{eqnarray*}
	\frac{\partial^2 \widetilde{\mathcal{L}}(\theta;S)}{\partial W^2}
	&=&
	- \frac{1}{n} \sum_{i=1}^n \Bigg[
	 \delta_i \times \frac{1}{(A_i - B_i)^2}
	\left( \frac{\partial B_i}{\partial W} \right)
	\left( \frac{\partial B_i}{\partial W} \right)^T
	+ \delta_i \times \frac{1}{A_i - B_i} \times \frac{\partial^2 B_i}{\partial W^2} \nonumber \\
	&& \qquad
	+ (1 - \delta_i) \times \left(
	- \frac{1}{B_i^2}
	\left( \frac{\partial B_i}{\partial W} \right)
	\left( \frac{\partial B_i}{\partial W} \right)^T
	+ \frac{1}{B_i} \times \frac{\partial^2 B_i}{\partial W^2}
	\right)
	\Bigg].
\end{eqnarray*}
The second partial derivative w.r.t. $V$ is
\begin{eqnarray*}
	\frac{\partial^2 \widetilde{\mathcal{L}}(\theta;S)}{\partial V^2}
	&=&
	- \frac{1}{n} \sum_{i=1}^n \Bigg[
	 \delta_i \times \frac{1}{(A_i - B_i)^2}
	\left( \frac{\partial B_i}{\partial V} \right)
	\left( \frac{\partial B_i}{\partial V} \right)^T
	+ \delta_i \times \frac{1}{A_i - B_i} \times \frac{\partial^2 B_i}{\partial V^2} \nonumber \\
	&& \qquad
	+ (1 - \delta_i) \times \left(
	- \frac{1}{B_i^2}
	\left( \frac{\partial B_i}{\partial V} \right)
	\left( \frac{\partial B_i}{\partial V} \right)^T
	+ \frac{1}{B_i} \times \frac{\partial^2 B_i}{\partial V^2}
	\right)
	\Bigg].
\end{eqnarray*}
The cross partial derivative is
\begin{eqnarray*}
	\frac{\partial^2 \widetilde{\mathcal{L}}(\theta;S)}{\partial V \partial W}
	&=&
	- \frac{1}{n} \sum_{i=1}^n \Bigg[
	 \delta_i \times \frac{1}{(A_i - B_i)^2}
	\left( \frac{\partial B_i}{\partial V} \right)
	\left( \frac{\partial B_i}{\partial W} \right)^T
	+ \delta_i \times \frac{1}{A_i - B_i} \times \frac{\partial^2 B_i}{\partial V \partial W} \nonumber \\
	&& \qquad
	+ (1 - \delta_i) \times \left(
	- \frac{1}{B_i^2}
	\left( \frac{\partial B_i}{\partial V} \right)
	\left( \frac{\partial B_i}{\partial W} \right)^T
	+ \frac{1}{B_i} \times \frac{\partial^2 B_i}{\partial V \partial W}
	\right)
	\Bigg],
\end{eqnarray*}
and
\begin{equation*}
	\frac{\partial^2 \widetilde{\mathcal{L}}(\theta;S)}{\partial W \partial V} =
	\left( \frac{\partial^2 \widetilde{\mathcal{L}}(\theta;S)}{\partial V \partial W} \right)^T.
\end{equation*}
In the hessian terms above we need the expression of different second-order derivatives of $B_i$ satisfying
\begin{eqnarray*}
	\frac{\partial^2 B_i}{\partial V^2}
	=
	2\, 	\overline{F}_{3,W}^d(y_i) \times \int_{y_i}^b \boldsymbol{\xi}(u,y_i) \boldsymbol{\xi}(u,y_i)^T \, du,
\end{eqnarray*}

\begin{eqnarray*}	
	\frac{\partial^2 B_i}{\partial W^2}
	=
		2 \left( \int_{y_i}^{b} f^{m,p}_{1,2,V}(u, y_i) \, du \right) \, \times M(y_i),
\end{eqnarray*}
	
\begin{eqnarray*}	
	\frac{\partial^2 B_i}{\partial V \partial W}
	=
	2 \left( \int_{y_i}^b V^T \boldsymbol{\xi}(u, y_i) \boldsymbol{\xi}(u,y_i) \, du \right) \times \left( M(y_i) \times W \right)^T.
\end{eqnarray*}
\noindent According to our case the loss function to be minimized is the negative log-likelihood $ \widetilde{l}(\theta;S) = - \widetilde{\mathcal{L}}_{m,p,d}(\theta;S)$. Algorithm~\ref{alg:snr-p} allows us to find the iterated maximum likelihood estimation. 
\begin{algorithm}
\caption{Maximum likelihood estimation.}
\label{alg:snr-p}
\begin{algorithmic}[1]
\State Initialize $(V_0, W_0, \dot{V}_0, \dot{W}_0)$ with $||V_0||_2 = 1$, $||W_0||_2 = 1$, $\big<V_0,\dot{V}_0 \big>_2=0$ and $\big<W_0,\dot{W}_0 \big>_2=0$.
\While{not converged}
\State Compute gradients: $F_V = \frac{\partial \widetilde{l}(\theta;S)}{\partial V}$, $F_W = \frac{\partial \widetilde{l}(\theta;S)}{\partial W}$.
\State Compute Hessians:
\[F_{VV} = \frac{\partial^2 \widetilde{l}(\theta;S)}{\partial V^2}, \quad F_{WW} = \frac{\partial^2 \widetilde{l}(\theta;S)}{\partial W^2}, \quad F_{VW} = \frac{\partial^2 \widetilde{l}(\theta;S)}{\partial V \partial W}\].
\State Form Hessian block matrix:
\[\text{Hess}F = 
\begin{pmatrix}
			F_{VV} & F_{VW} \\
			F_{VW}^T & F_{WW}
		\end{pmatrix}
		-
		\begin{pmatrix}
			F_V^\top V I_{mp} & 0 \\
			0 & F_W^\top W I_d
		\end{pmatrix}
		\]
		\State Solve Newton system for velocity update:
		\[
		\begin{pmatrix}
			\dot{V} \\
			\dot{W}
		\end{pmatrix}
		\gets
		\begin{pmatrix}
			\dot{V} \\
			\dot{W}
		\end{pmatrix}
		- \epsilon
		 \begin{pmatrix}
			I_{mp} - V V^\top & 0 \\
			0 & I_d - W W^\top
		\end{pmatrix}
		\text{Hess}F^{-1}
		\begin{pmatrix}
			I_{mp} - V V^\top & 0 \\
			0 & I_d - W W^\top
		\end{pmatrix}
		\begin{pmatrix}
			F_V \\
			F_W
		\end{pmatrix}
		\]
		\State Retraction on the spheres (via geodesics):
		\[
		V \leftarrow V \cos(\epsilon ||\dot{V}||_2) + \frac{\dot{V}}{||\dot{V}||_2} \sin(\epsilon ||			\dot{V}||_2)
		\]
		\[
		W \leftarrow W \cos(\epsilon ||\dot{W}||_2) + \frac{\dot{W}}{||\dot{W}||_2} \sin(\epsilon ||\dot{W}||_2)
		\]
		\EndWhile
	\end{algorithmic}
\end{algorithm}
\newpage
%\textcolor{red}{The maximization of the log-likelihood  (\ref{estloglikelihood}) for fixed degrees   $m$, $p$ and $d$ yields the maximum likelihood estimates of the model parameters $W$ and $V$:...................................}
%$$\widetilde{\theta}_n=\left(\widetilde{W}_n,\widetilde{V}_n\right) =\underset{\theta \in \mathbb{R}^{mp+d}}{\argmax} \mathcal{L}(\theta;S).$$
\subsection{Estimation of  the survival functions}
We now consider the  estimation of the survival functions $\overline{F}_T$ and $\overline{F}_C$, as well as the joint survival $P$ of the survival time $T$ and the censoring time $C$. For this, we assume that we have an independent and identically distributed (\textit{i.i.d.}) sample $(Y_i,\delta_i)_{1\leq i \leq n}$ of size 
$n$, drawn from the observed model $(Y,\delta)$.
Based on the estimator $\widetilde{\theta}_n=\left(\widetilde{W}_n,\widetilde{V}_n\right)$,
 we  obtain the estimators of the joint survival function $\overline{F}_{1,2}$ and the joint survival function of $\overline{F}_3$, namely,
\begin{eqnarray*}
 \widetilde{\overline{F}}_{1,2,V}^{m,p,n}(x,y)&=& \int_x^{a} \int_y^{b}  \widetilde{f}_{1,2,V}^{m,p,n}(x,y) dx dy
 \\&=& \int_x^{a} \int_y^{b} \widetilde{V}_n^T \boldsymbol{\xi}(x,y) \boldsymbol{\xi}(x,y)^T \widetilde{V}_n dxdy,     
\end{eqnarray*}
and 
$$\widetilde{\overline{F}}_{3,W}^{d,n}(t)=\widetilde{W}_n^T \Big( \int_{t}^{b} \Phi(x) \Phi(x)^T dx  \Big) \widetilde{W}_n.$$
The estimation of $\overline{F}_1(x)$  and $\overline{F}_2(y)$ are given respectively by
$$ \widetilde{\overline{F}}_{1,V}^{m,p,n}(t)=\widetilde{\overline{F}}_{1,2,V}^{m,p,n}(t,b) \quad \text{and} \quad \widetilde{\overline{F}}_{2,V}^{m,p,n}(t)=\widetilde{\overline{F}}_{1,2,V}^{m,p,n}(a,t). $$
We conclude that the estimation of $\overline{F}_T$ and $\overline{F}_C$ are given by 
$$\widetilde{\overline{F}}_{T}^{m,p,d,n}(t)= \widetilde{\overline{F}}_{1,V}^{m,p,n}(t) \widetilde{\overline{F}}_{3,W}^{d,n}(t) \quad \text{ and } \quad  \widetilde{\overline{F}}_{C}^{m,p,d,n}(t)=  \widetilde{\overline{F}}_{2,V}^{m,p,n}(t) \widetilde{\overline{F}}_{3,W}^{d,n}(t).$$
Finally, we obtain the estimator of the survival joint distribution   $\widetilde{P}$,  given by
$$\widetilde{P}^{m,p,d,n}(t,s)=
\min \left(\widetilde{\overline{F}}_{C}^{m,p,d,n}(s)\widetilde{\overline{F}}_{T}^{m,p,d,n}(t)^{1-\widetilde{\alpha}_n^{m,p,d}(t)}, \widetilde{\overline{F}}_{T}^{m,p,d,n}(t) \widetilde{\overline{F}}_{C}^{m,p,d,n}(s)^{1-\widetilde{\beta}_n^{m,p,d}(s)} \right) 
\frac{\overline{F}_{1,2,V}^{m,p,d,n}(x,y)}{\overline{F}_{1,V}^{m,p,d,n}(x) \overline{F}_{2,V}^{m,p,d,n}(y)},
$$
where $$ \widetilde{\alpha}_n^{m,p,d}(t)=\frac{-\log \left(\widetilde{\overline{F}}_{3,W}^{d,n}(t)\right)}{-\log\left(\widetilde{\overline{F}}_{T}^{m,p,d,n}(t)\right) }\quad \text{and} \quad \widetilde{\beta}_n^{m,p,d}(t)=\frac{-\log \left(\widetilde{\overline{F}}_{3,W}^{d,n}(t)\right)}{-\log\left(\widetilde{\overline{F}}_{C}^{m,p,d,n}(t)\right) }.$$

\subsection{Asymptotic properties}
We start by establishing the asymptotic properties of  $\theta_n$, we rely on the framework provided by 
\cite{Whi82} which develops sufficient conditions for the asymptotic normality and consistency  of the ML estimators. To do so, 
let 
\begin{eqnarray*}
g(t,\ell;f_{1,2},f_3)&=&\left[f_Y(t)-f_{Y,\delta=0}(t) \right]^{\ell} \left[f_{Y,\delta=0}(t) \right]^{1-\ell}, 
\end{eqnarray*}

\begin{eqnarray*}
h(t,\ell; \theta)
&=&  \left[\widetilde{f}_{Y,n}(t)-\widetilde{f}_{Y,\delta=0}^{m,p,d}(t) \right]^{\ell} \left[\widetilde{f}_{Y,\delta=0}^{m,p,d}(t)  \right]^{1-\ell},
\end{eqnarray*}
where $\widetilde{f}_{Y,\delta=0}^{m,p,d}(t) $ is given by Equation (\ref{estfydelta=0}).
The Kullback-Leibler information criterion given by 
$$\displaystyle{\text{KL}(g||h;\theta)=\mathbb{E} \left( \log \left[ \frac{g(Y,\delta;f_{1,2},f_3)}{h(Y,\delta;\theta) } \right]\right)}.$$ 
Consider the following assumptions:
\begin{enumerate}
\item[ ] $\mathbf{(H_1})$:  The functions $t \to \int_t^{+\infty} f_{1,2} (x,t)dx$, $t \to \int_t^{+\infty} f_{1,2} (t,y)dy$ and $t \to f_3(t)$ belongs to $L^2([0,+\infty[)$.  
\item[ ] $\mathbf{(H_2)}$: 
(a) $\theta$ is interior to $\mathbb{R}_+^{m+2}$.
(b) $B(\theta)$ is not singular.
(c) $\theta$ is a regular point of $A(\theta)$. 
\end{enumerate}
These assumptions constitute standard regularity conditions used in maximum likelihood analysis. 
Assumption $\mathbf{(H_1)}$ ensures appropriate integrability properties, which are necessary to control the variability of the likelihood and to justify the use of limit theorems.
Assumption $\mathbf{(H_2)}$ serves as regularity conditions commonly adopted in maximum likelihood analysis to demonstrate asymptotic properties. The following proposition shows the existence of the ML estimator $\theta_n$.

\begin{prop}[Existence]\label{existance}
For all $n \in \mathbb{N}^{\star}$, there exists a measurable ML estimator $\theta_n$.
\end{prop}

\noindent
Once the existence of the maximum likelihood estimator is guaranteed, the next step is to study its asymptotic properties, in particular its consistency.

\begin{prop}[Consistency]\label{consistance}
The estimator $\theta_n$ converges almost surely to the true parameter $\theta$:
$\theta_n \xrightarrow{\text{a.s.}} \theta \quad \text{as } n \to \infty$,
for almost every sequence $(Y_i,\delta_i)_{i\geq 1}$ of censored observations.
\end{prop}

\noindent
Having established consistency, we now proceed to refine the analysis of the estimator’s asymptotic behavior.

\noindent
While Proposition \ref{existance} establishes the existence of a ML estimator, uniqueness must still be verified. This allows us, in the next proposition, to show the identifiability of $\theta$ and the asymptotic normality for the parameter vector $\theta_n$.

\begin{prop}[Asymptotic Normality]\label{aymnormal}
Under Assumptions $(\mathbf{H_1})-(\mathbf{H_2})$ and for a fixed $m,p,d \geq 0$, $\theta$ is identified and 
$$\sqrt{n} \left(\theta_n- \theta \right) \underset{n \to +\infty}{\overset{\mathcal{L}}{\longrightarrow}} \mathcal{N} \left(0,C \left(\theta\right) \right),$$
\end{prop}

\noindent
The proof of Proposition \ref{aymnormal} is done by verifying the conditions of Theorem 3.2 in \cite{Whi82}.

\section{Experimental results}
\label{sec:sim}
We now assess the performance of the proposed method using both simulated and real datasets.
\subsection{Simulation study}
We firstly investigate the performance of the proposed estimators within a simulation study. To this end, we perform $N = 100$ Monte Carlo simulations each with  three sample sizes $n = 50, 100,$ and $200$. The random variables $X_i \sim \mathcal{E}(a_i)$, for $i = 1, 2, 3$, are generated from exponential distributions with respective rate parameters $a_1 = 2$, $a_2 = 1.5$, and $a_3 = 3$.
For the bivariate survival distribution associated with $(X_1, X_2)$, we examine the following settings:
\begin{enumerate}
\item[$\bullet$] Model (i): The dependence structure between $X_1$ and $X_2$ is defined via a Clayton survival copula: $
C_\theta(u,v)
=
\left( \max \big \{
u^{-\theta}
+
v^{-\theta}
-
1 ; 0 \big \}
\right)^{-1/\theta}
$, with a parameter $\theta$ fixed at 4.
\item[$\bullet$] Model (ii): The dependence structure between $X_1$ and $X_2$ is defined via a Gumbel survival copula: $
C_\theta(u,v)
=
\exp\!\left(
-
\left[
(-\ln u)^\theta
+
(-\ln v)^\theta
\right]^{1/\theta}
\right)
$, with a parameter $\theta$ fixed at 4.
\end{enumerate}
An example with \(n = 50\) samples and Clayton copula is illustrated in Figure~\ref{fig:sim_data}. For visualization purposes, only 10 simulated trajectories of \(X_{1i}\), \(X_{2i}\), and \(X_{3i}\) are displayed in the top-left panel. From these variables, we define
$
T_i = \min(X_{1i}, X_{3i})
\quad \text{and} \quad
C_i = \min(X_{2i}, X_{3i})
$
shown in the top-right panel. The observed data are then obtained as
$
Y_i = \min(T_i, C_i)
$
and represented in the bottom-left panel. Finally, the bottom-right panel displays the kernel density estimate (KDE) computed from the observed data \(Y_i\).

\begin{figure}[h!]
    \centering
    \includegraphics[width=1\textwidth]{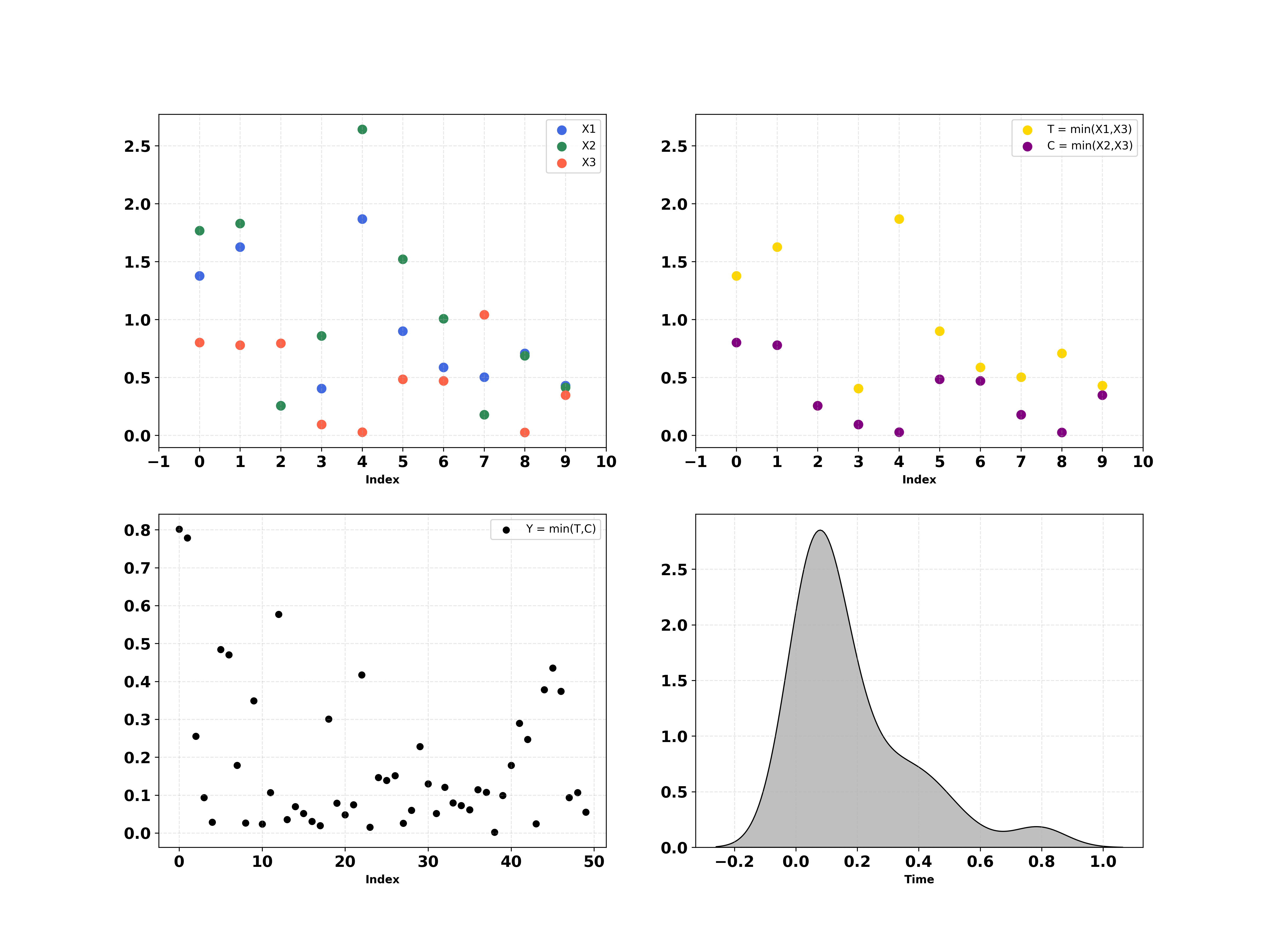} 
    \caption{Illustration of the simulated data generation process with \(n=50\) samples: $10$ simulated variables \((X_{1i},X_{2i},X_{3i})\) (top-left), construction of \((T_i,C_i)\) (top-right), all observed data \(Y_i\) (bottom-left), and the corresponding kernel density estimate from \(Y_i\) (bottom-right).}
    \label{fig:sim_data}
\end{figure}

In our experiments, we consider Laguerre polynomials which are appropriate because the lifetime variable is supported on the semi-infinite interval $I=[0, + \infty[$. Laguerre basis functions, defined by $\phi_l(t) = \frac{e^t}{l!}\frac{d^l}{dt^l}\left(e^{-t}t^l\right)$, naturally  provide an orthogonal decomposition with respect to the exponential weight function $\rho(t) = e^{-t}$. This leads to stable numerical estimation and efficient approximation of  survival functions. This is confirmed by the simulations, where alternative polynomial bases such as Legendre, Hermite, and Chebyshev yielded less accurate estimates, highlighting their inability to adequately control the structural constraints imposed by the bivariate density and the univariate survival function. In contrast, the Laguerre decomposition provides a natural compatibility with these constraints, suggesting that it is the most effective choice among the non-exhaustive set of polynomial bases considered.

Figure~\ref{fig:poly} illustrates the first five basis polynomials $\phi_l(t)$ for $l = 1, \dots, 5$, starting from 1 at $t=0$ and decaying to 0 as $t \to \infty$, with increasingly oscillatory behavior as the polynomial degree increases.
Figure~\ref{fig:survie-pdf} shows, on the left, a valid univariate survival function $\overline{F}_{W}^{5}(x)$ for $x \in [0,10]$, reconstructed from a vector of spherical random coefficients $W$ satisfying $\sum_{l=1}^{5}w_l^2 = 1 $. On the right, it shows a valid bivariate probability density function $f_{V}^{5,5}(x_1,x_2)$ for $(x_1,x_2) \in [0,10]^2$, reconstructed from a matrix of spherical random coefficients $V$ satisfying $\sum_{l_1,l_2=1}^{5}v_{l_1,l_2}^2 = 1 $.

\begin{figure}
    \centering
    \includegraphics[width=.6\textwidth]{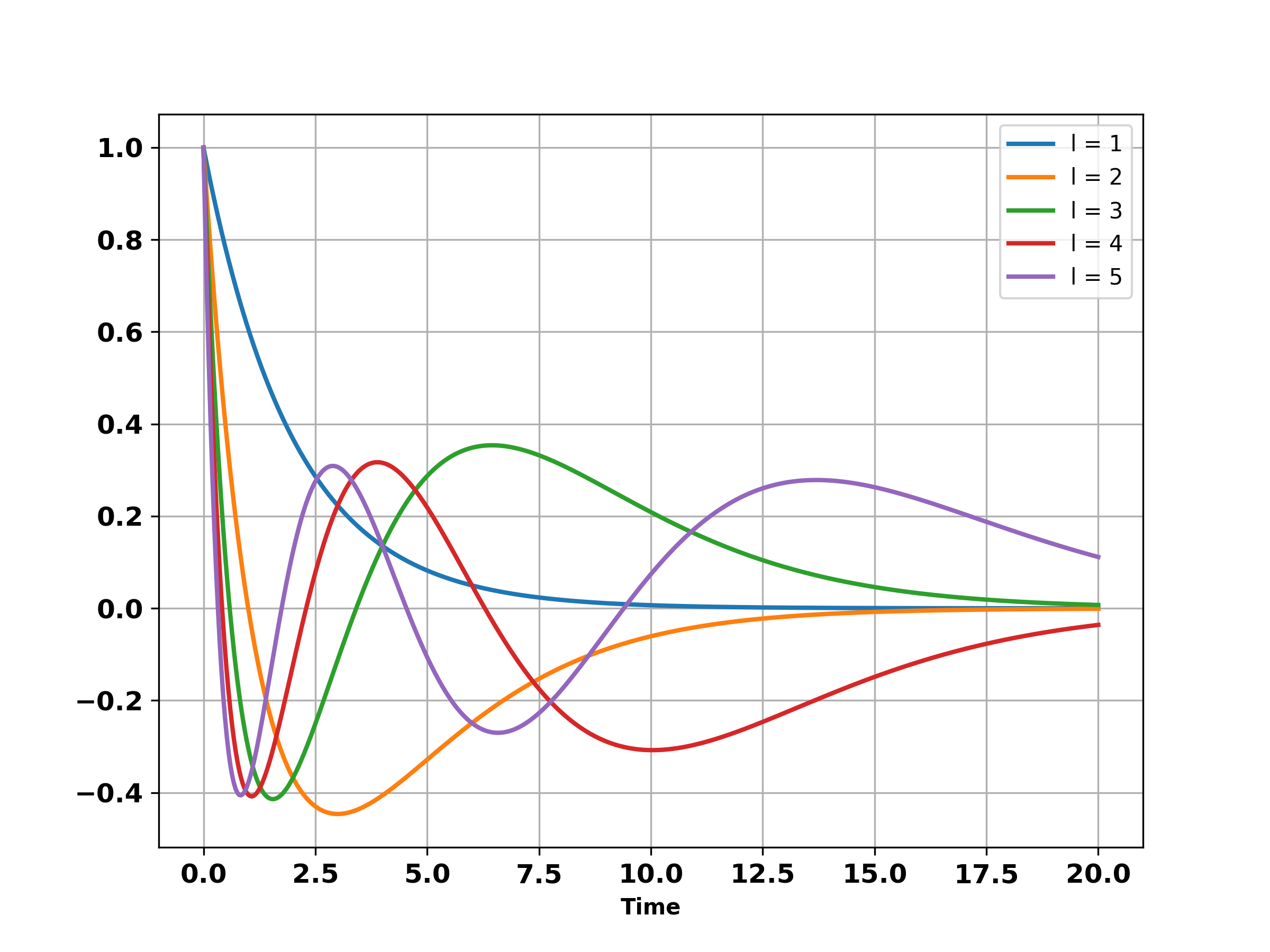} 
    \caption{First examples of Laguerre basis polynomials $\phi_l(t)$ for $l = 1, \dots, 5$.}
    \label{fig:poly}
\end{figure}

We then perform numerical maximization of the log‑likelihood function presented in Algorithms~\ref{alg:snr-p}. To select appropriate Laguerre  polynomial degrees, we fit models across various combinations of $(m, p)$, where $m, p \in \{1, 2, 3, 4, 5, 6\}$. In addition, polynomial degrees $d$ ranging from 1 to 10 are also considered. Model selection among these configurations is carried out using the AIC criterion.  Specifically, we obtain the parameter estimates for a selection of triplets  $(m,p,d)$ and compute $$\text{AIC}(m,p,d)=2(m+p+d+1)-2 \widetilde{\mathcal{L}}_{m,p,d}(\theta;S),$$ for
each triplet.
Figure \ref{fig:nll} shows that the negative log-likelihood $ - \widetilde{\mathcal{L}}_{m,p,d} $ reaches a local minimum after only a few iterations of Algorithm~\ref{alg:snr-p}, indicating fast convergence. In particular, convergence is achieved after 24 iterations for the Clayton copula and 20 iterations for the Gumbel copula.

\begin{figure}
    \centering
    \includegraphics[width=.45\textwidth]{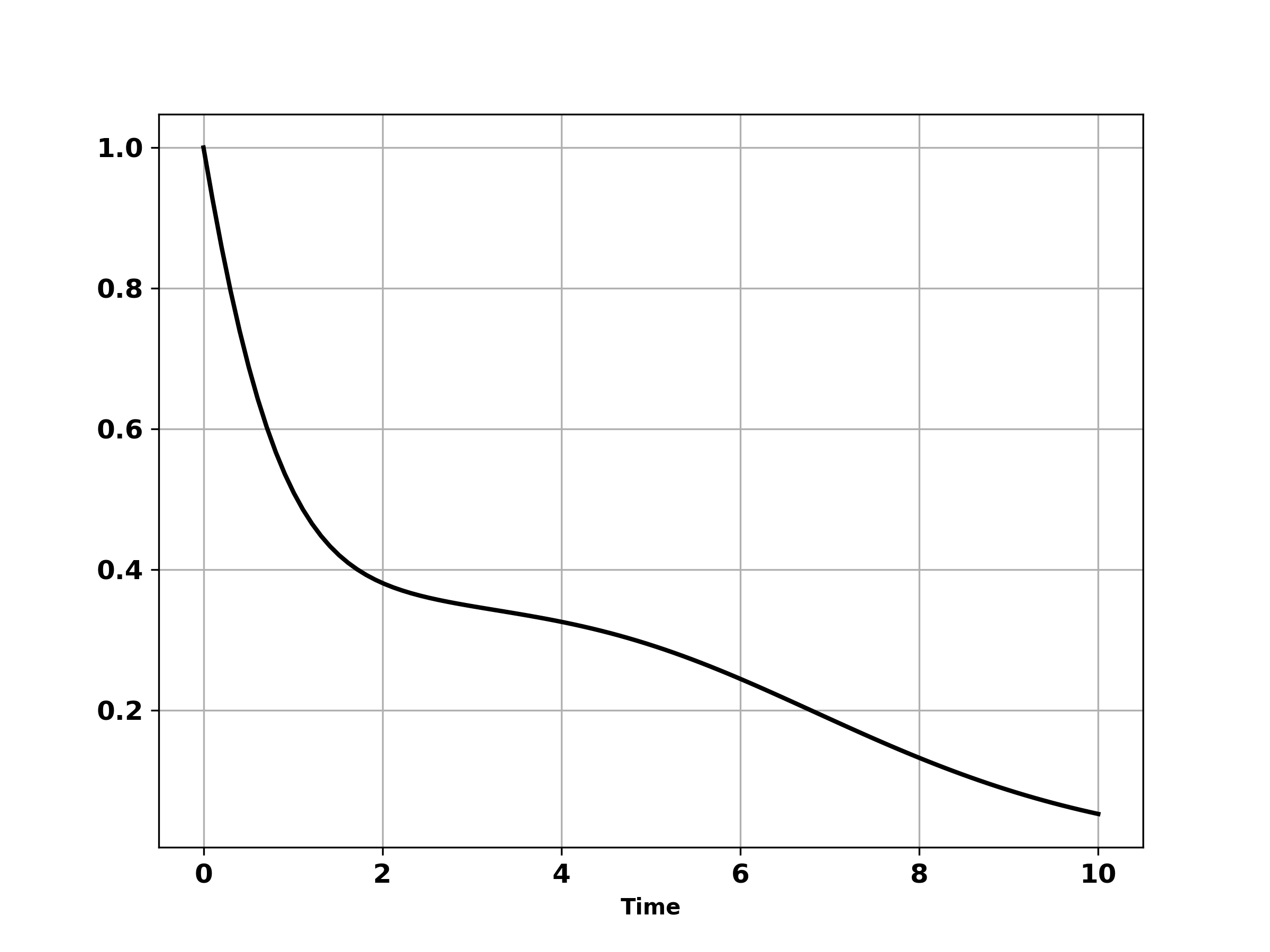}
    \includegraphics[width=.54\textwidth]{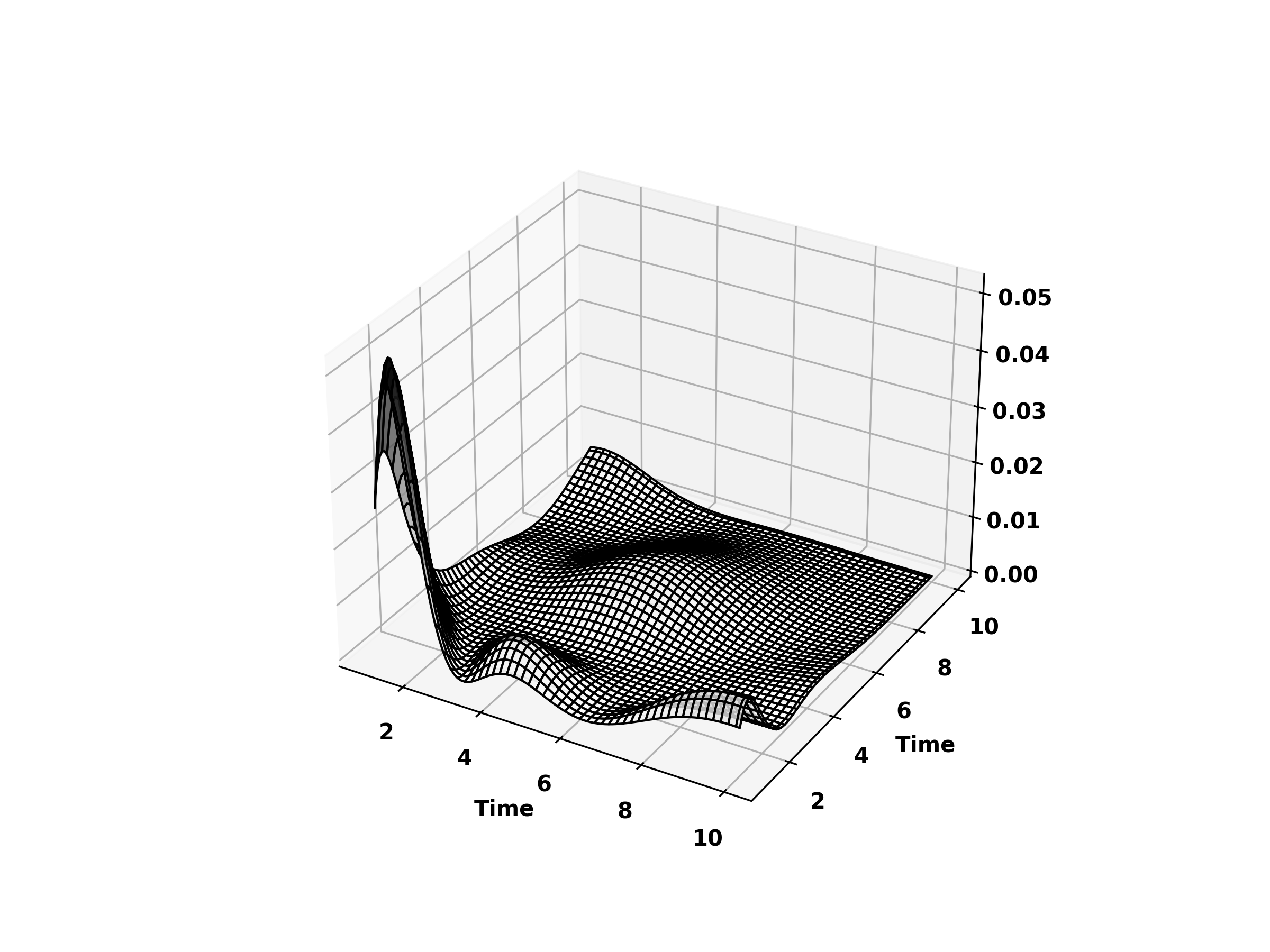}
    \caption{Reconstructed univariate survival function (left) and bivariate density function (right) from spherical random coefficients.}
    \label{fig:survie-pdf}
\end{figure}

Once the optimal degrees have been selected and the corresponding parameters estimated, we evaluate the performance of the proposed approach in approximating the joint survival function of $(T, C)$ as well as the marginal survival functions of $T$ and $C$. 
To assess the performance of the proposed estimators, we compute their bias, empirical standard deviation (SD), and mean squared error (MSE). We consider the Fréchet mean of estimated parameters obtained during the $N=100$ Monte-Carlo replications as a mean estimator giving the functional estimators of marginal and joint survival functions. 
In addition, The coverage probabilities (CP) were calculated as the proportion of Monte-Carlo iterations in which the true survival functions  fell within the empirical  confidence intervals
$$
\mathrm{CI}_{95\%}\big(\overline{F}_{n}\big) = \overline{F}_{n} \pm z_{0.975} \, \mathrm{SD}\big(\overline{F}_{n}\big),
$$
where $F_n$ is the  functional of interest  (e.g., \(\widetilde{\overline{F}}_T(t)\), \(\widetilde{\overline{F}}_C(t)\), or \(\widetilde{P}(t,s)\)).

\begin{figure}
    \centering
    \includegraphics[width=.6\textwidth]{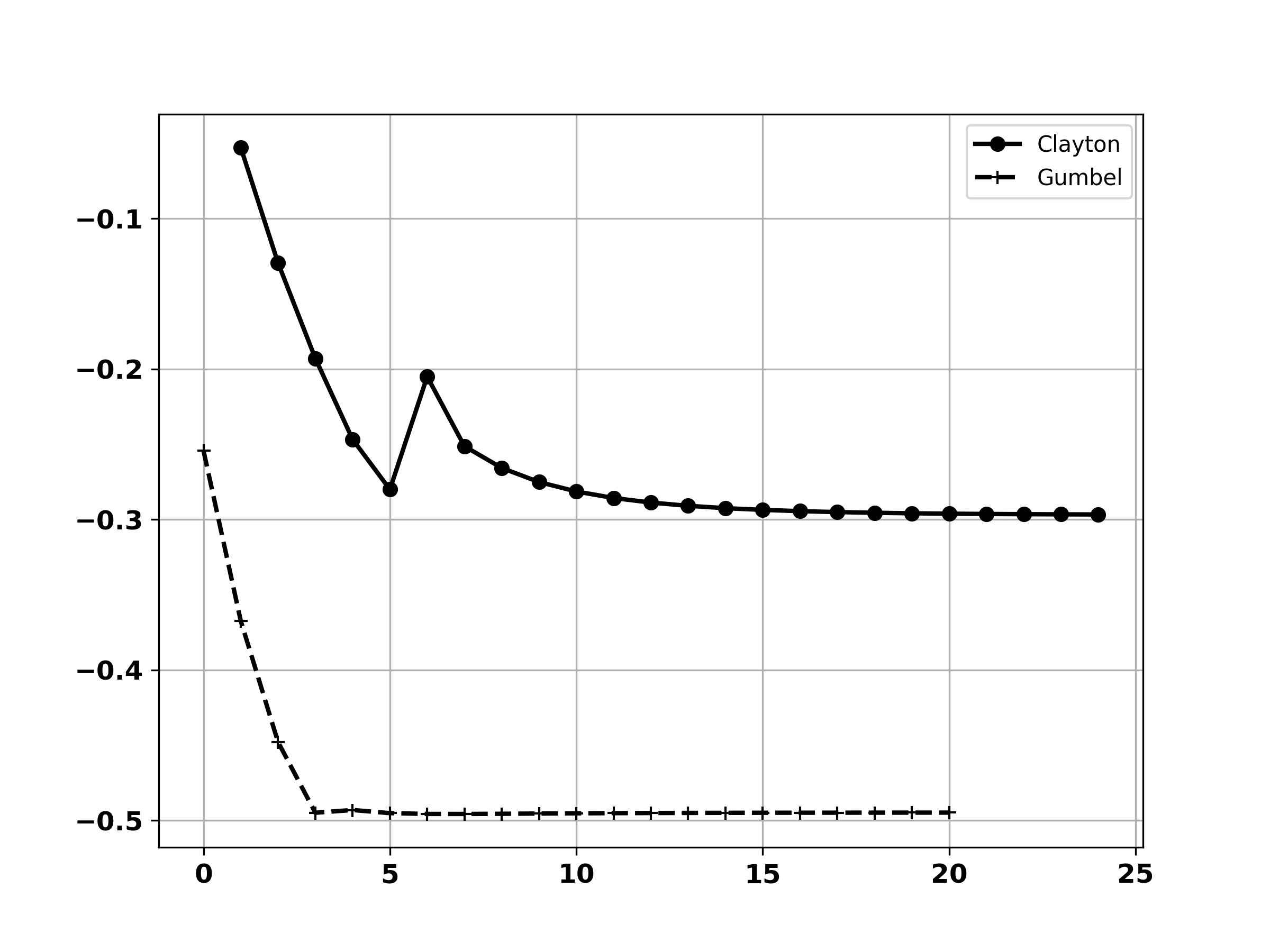} 
    \caption{Loss as the negative log-likelihood as a function of the number of iterations.}
    \label{fig:nll}
\end{figure}

We firstly illustrate the estimation of the marginal survival functions. Figures \ref{fig:sim_marginal_clayton1} and \ref{fig:sim_marginal_gumbel1} display the results of $N=100$ simulations under models (i) and (ii), respectively. The left panels show the individual estimates, while the right panels present their Fréchet mean, which serves as our proposed estimator (in gray), together with the true survival functions (in black). In each figure, the top row corresponds to $\overline{F}_T$ and the bottom row to $\overline{F}_C$. For comparison, the Kaplan--Meier estimator is also included (in blue) to highlight its behavior under dependence. These results indicate that our approach better preserves the dependence structure while providing estimates that closely match the true marginal survival functions.
Figures \ref{fig:sim_marginal_clayton2} and \ref{fig:sim_marginal_gumbel2}  display the  results of the joint survival  estimation.

Tables \ref{tab:metric_clayton}--\ref{tab:metric_gumbel} report the Monte Carlo bias, empirical standard deviation (SD), mean squared error (MSE), and empirical coverage probability (CP) of the proposed functional estimators.
The results show a clear reduction in bias as the sample size increases, supporting the consistency of the estimators. Moreover, the empirical SD decreases with \(n\), indicating improved stability for larger samples. The MSE, which is mainly driven by the variance, also decreases and tends to zero as \(n \to \infty\). Finally, the empirical coverage probability improves with increasing sample size, approaching the nominal confidence level.

\begin{figure}
    \centering
    \includegraphics[width=1\textwidth]{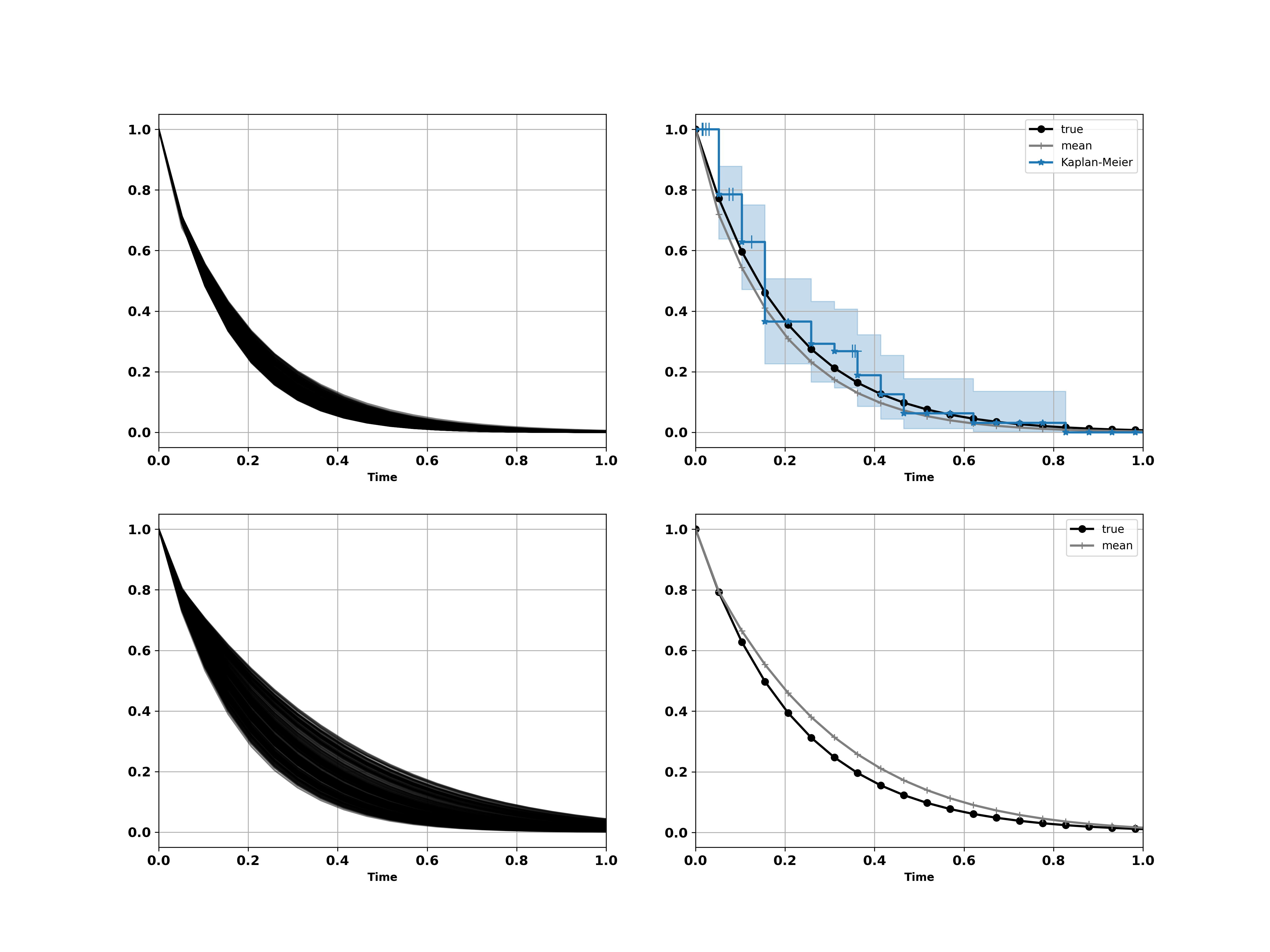} 
    \caption{Estimated marginal survival functions from 100 simulations (left) and their Fréchet mean in red compared to the true function in blue (right) under the Clayton copula: top row for $T$ and bottom row for $C$.}
    \label{fig:sim_marginal_clayton1}
\end{figure}

\begin{figure}
    \centering
    \includegraphics[width=1\textwidth]{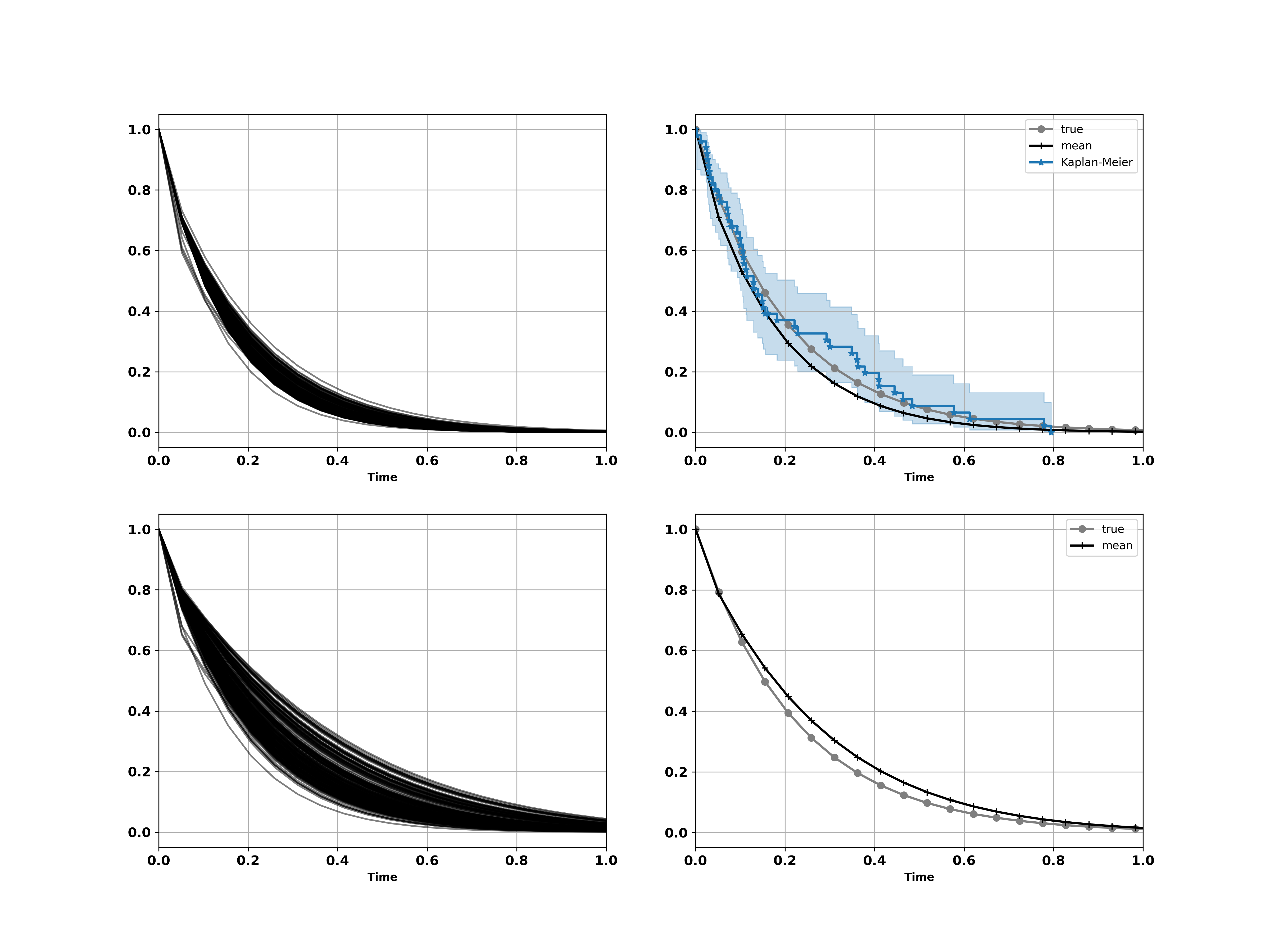} 
    \caption{Estimated marginal survival functions from 100 simulations (left) and their Fréchet mean in red compared to the true function in blue (right) under the Gumbel copula: top row for $T$ and bottom row for $C$.}
    \label{fig:sim_marginal_gumbel1}
\end{figure}

\begin{figure}
    \centering
    \includegraphics[width=1\textwidth]{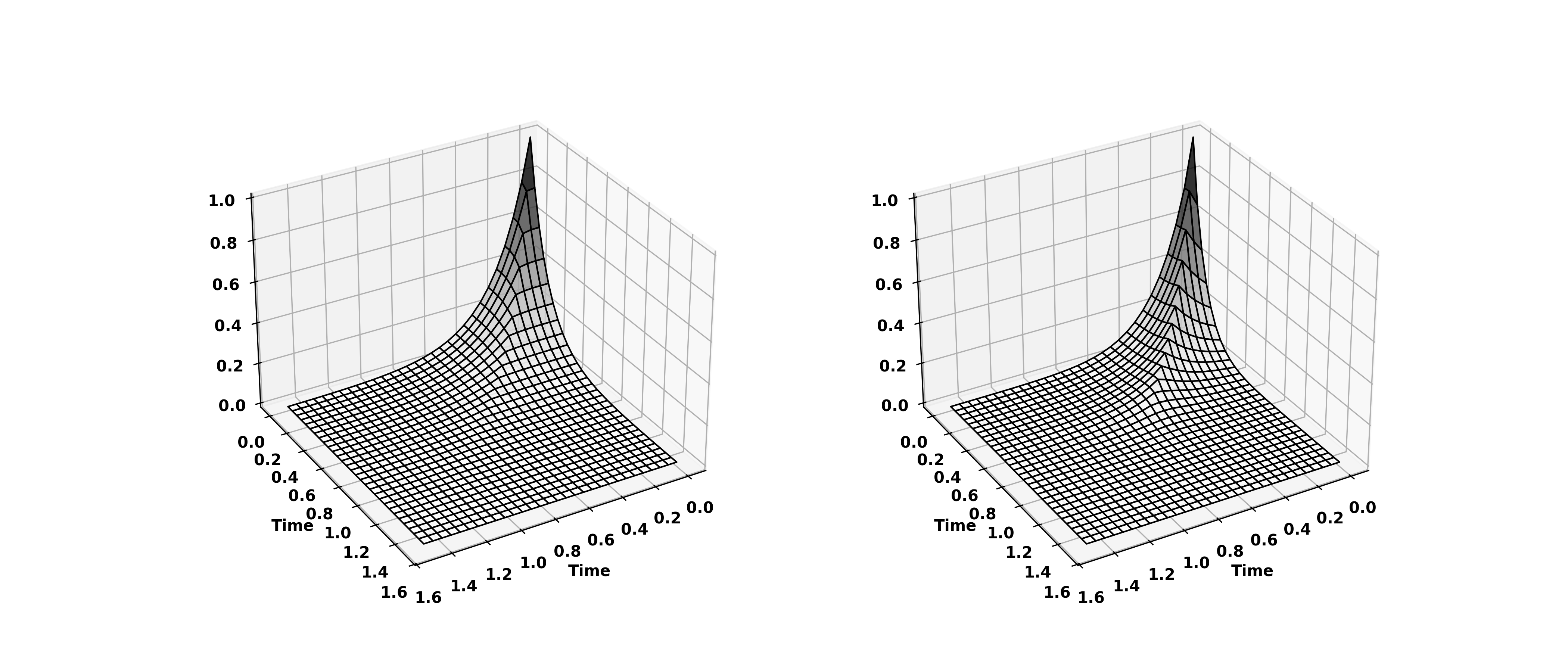} 
    \caption{True joint survival function of $(T,C)$ (left) and the Fréchet mean of estimated joint survival functions from 100 simulations (right) under the Clayton copula.}
    \label{fig:sim_marginal_clayton2}
\end{figure}

\begin{figure}
    \centering
    \includegraphics[width=1\textwidth]{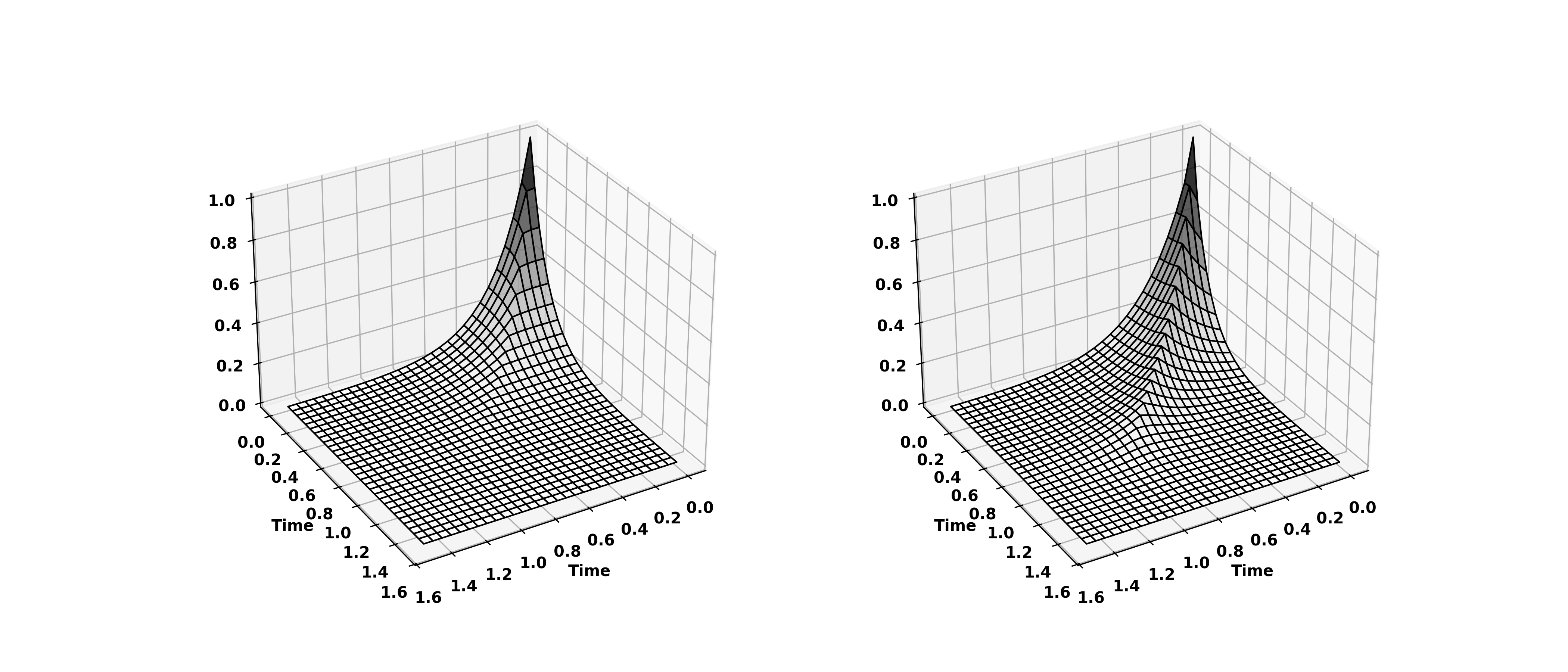} 
    \caption{True joint survival function of $(T,C)$ (left) and the Fréchet mean of estimated joint survival functions from 100 simulations (right) under the Gumbel copula.}
    \label{fig:sim_marginal_gumbel2}
\end{figure}

\begin{table}
\centering
\begin{tabular}{|c|c|c|c|c|c|c|c|c|c|}
\hline
\diagbox{Metric}{Survival} & \multicolumn{3}{c|}{$n=50$} & \multicolumn{3}{c|}{$n=100$}
& \multicolumn{3}{c|}{$n=200$}\\
\cline{2-10}
 & $T$ & $C$ & $(T,C)$ & $T$ & $C$ & $(T,C)$ & $T$ & $C$ & $(T,C)$ \\
\hline
Bias & 0.029 & 0.042  & -0.0322  & -0.022 & 0.038  & 0.0061  & 0.017 & 0.0167 & 0.006 \\
\hline
MSE & 0.0018 & 0.0044 & 0.004  & 0.001  & 0.003 & 0.00048   & 0.0005 & 0.0006 & 0.0003 \\
\hline 
SD & 0.031  & 0.0525 & 0.054  &  0.024  & 0.044 & 0.021 & 0.014  & 0.0182 & 0.017 \\
\hline
CP & 96.6 $\%$  & 90 $\%$  & 93.3 $\%$  & 97.5 $\%$ & 91.8 $\%$ & 94.7 $\%$ & 98.3 $\%$ & 93.4 $\%$ & 95.9 $\%$ \\
\hline
\end{tabular}
\caption{Results of survival estimations with Clayton copula.}
\label{tab:metric_clayton}
\end{table}

\begin{table}
\centering
\begin{tabular}{|c|c|c|c|c|c|c|c|c|c|}
\hline
\diagbox{Metric}{Survival} & \multicolumn{3}{c|}{$n=50$} & \multicolumn{3}{c|}{$n=100$}
& \multicolumn{3}{c|}{$n=200$}\\
\cline{2-10}
 & $T$ & $C$ & $(T,C)$ & $T$ & $C$ & $(T,C)$ & $T$ & $C$ & $(T,C)$ \\
\hline
Bias & 0.03 & 0.028 & 0.01 & 0.026 & 0.002 & -0.0093 & 0.02 & 0.0063 & 0.0016 \\
\hline
MSE  & 0.003 & 0.0018  & 0.0005 & 0.0014 &  0.001 & 0.0009  & 0.00022 & 
$1.4 \times 10^{-5}$ & $5.69 x 10^{-5}$  \\
\hline 
SD  & 0.027 & 0.041  & 0.030  & 0.0041  & 0.0324  & 0.018  & 0.0032 & 0.028 & 0.014 \\
\hline
CP & 96.6 $\%$ & 96.6 $\%$  & 94.7 $\%$  & 97.4 $\%$ & 97.6 $\%$ & 95.2 $\%$ & 98 $\%$ & 98.1 $\%$ & 95.7 $\%$ \\
\hline 
\end{tabular}
\caption{Results of survival estimations with Gumbel copula.}
\label{tab:metric_gumbel}
\end{table}

\newpage
\subsection{Real data examples}
We consider two real datasets to further demonstrate the effectiveness of our method.
\subsubsection{Medical data}
In this study, we analyze data from the Diabetic Retinopathy Study (DRS), a clinical trial sponsored by the National Eye Institute that investigated whether laser photocoagulation could postpone the development of blindness among individuals affected by diabetic retinopathy. Further information regarding the design and findings of the trial is available in \cite{Cso89} and \cite{Fei15}.
The dataset includes 71 patients who met the study eligibility criteria. At enrollment, one eye of each participant was randomly assigned to receive laser therapy using either argon laser, xenon arc, or a combination of the two techniques, whereas the fellow eye served as an untreated control.

In this framework, we denote by $T$ the time until blindness for the eye receiving laser therapy, while $C$ represents the (potentially censored) time to blindness for the untreated eye. For each patient, the observed quantity is
$Y=\min(T,C)$, together with the indicator $\delta$ identifying the corresponding failure mechanism. In certain cases, one may observe $T=C$, indicating that both eyes can become blind at the same time.
All observation times are converted into years by dividing the recorded values (in days) by 365. As discussed in \cite{Bar24}, these data can be appropriately described using a Marshall--Olkin Weibull model, where three independent random variables $X_i \sim \mathcal{W}(k,\lambda_i)$ for $i=1,2,3$ are assumed, and the observed times are defined by $T=\min(X_1,X_3)$ and $C=\min(X_2,X_3)$.
Applying the maximum likelihood estimation procedure proposed in \cite{Fei15}, we obtain the estimates
$\widetilde{k}=1.558$, $\widetilde{\lambda}_1=0.184$, $\widetilde{\lambda}_2=0.224$, and $\widetilde{\lambda}_3=0.059$. 
Figure \ref{fig:medical} displays the estimated marginal survival functions of $T$. The resulting survival curves obtained from the Kaplan--Meier (KM) estimator and from our proposed method are very similar.
This suggests that, for this dataset, it is reasonable to work under the GMO framework, as shown in \cite{esc24}, where it is established that within the GMO model the KM estimator is a consistent estimator of the survival function of $T$. Moreover, the Bayesian estimator in \cite{pin15} appears to deviate more from our estimator. This discrepancy supports the idea that the Weibull parametric assumption may not be fully adequate for this dataset, whereas our approach is nonparametric and therefore more flexible.

\subsubsection{UEFA Champion's League}
 We consider a soccer dataset from the UEFA Champions League previously studied in \cite{Mei07}. The dataset is restricted to matches satisfying two conditions: (i) the home team scored at least one goal, and (ii) at least one goal in the match was scored directly from a kick action (penalty, free kick, or other type of direct shot) by either team.

Let $T$ denote the time (in minutes) of the first goal scored from a kick by either side, and let $C$ denote the time of the first goal scored by the home team regardless of its type. For such nonnegative continuous observations, three situations may occur: $\{T < C\}$, $\{T > C\}$, or $\{T = C\}$.
As discussed in \cite{pin15}, a Bayesian study implemented via OpenBUGS suggests that these observations can be adequately modeled using an EGMO structure with exponential marginals. This leads to the assumption of three exponential random variables $X_i \sim \mathcal{E}(\lambda_i)$ for $i=1,2,3$. The corresponding Bayesian estimates are given by $\widetilde{\lambda}_1=0.001$, $\widetilde{\lambda}_2=0.375$, and $\widetilde{\lambda}_3=0.04$, with the additional assumption that $(X_1,X_2)$ is independent of $X_3$.

The marginal survival function of $T$, displayed in Figure \ref{fig:foot}, shows a clear discrepancy between the KM estimator and the two alternative approaches. In particular, the KM curve differs notably from both our proposed estimator and the Bayesian estimator of \cite{pin15}, which remain closely aligned.
This lack of agreement raises concerns regarding the validity of the GMO modeling assumption discussed in \cite{esc24}. More generally, these results emphasize the importance of rigorous model checking and validation procedures in survival analysis with EGMO models.

\begin{figure}
    \centering
    \includegraphics[width=0.5\textwidth]{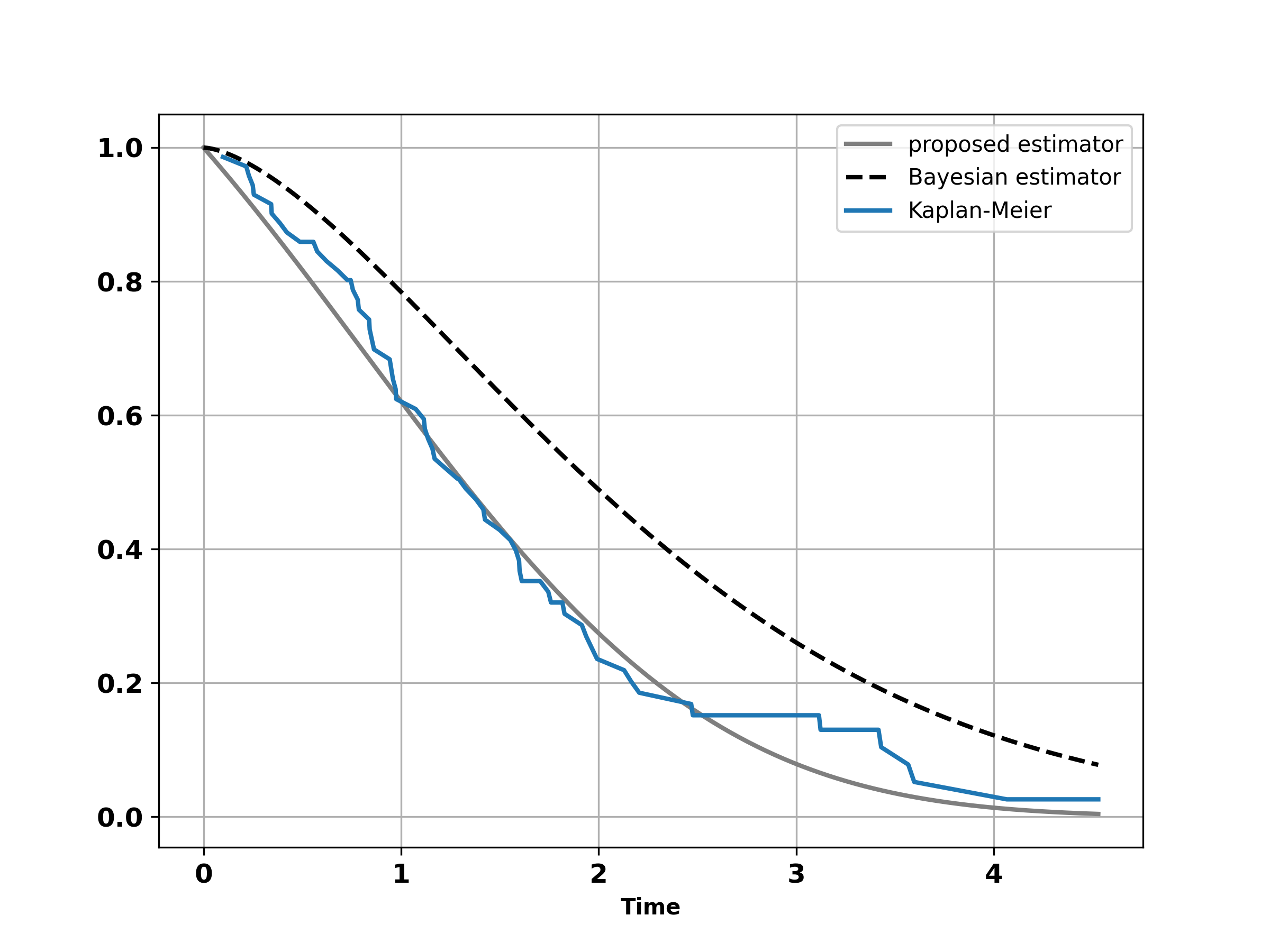} 
    \caption{Results for the survival function estimator of $T$ with the medical data.}
    \label{fig:medical}
\end{figure}

\begin{figure}
    \centering
    \includegraphics[width=0.5\textwidth]{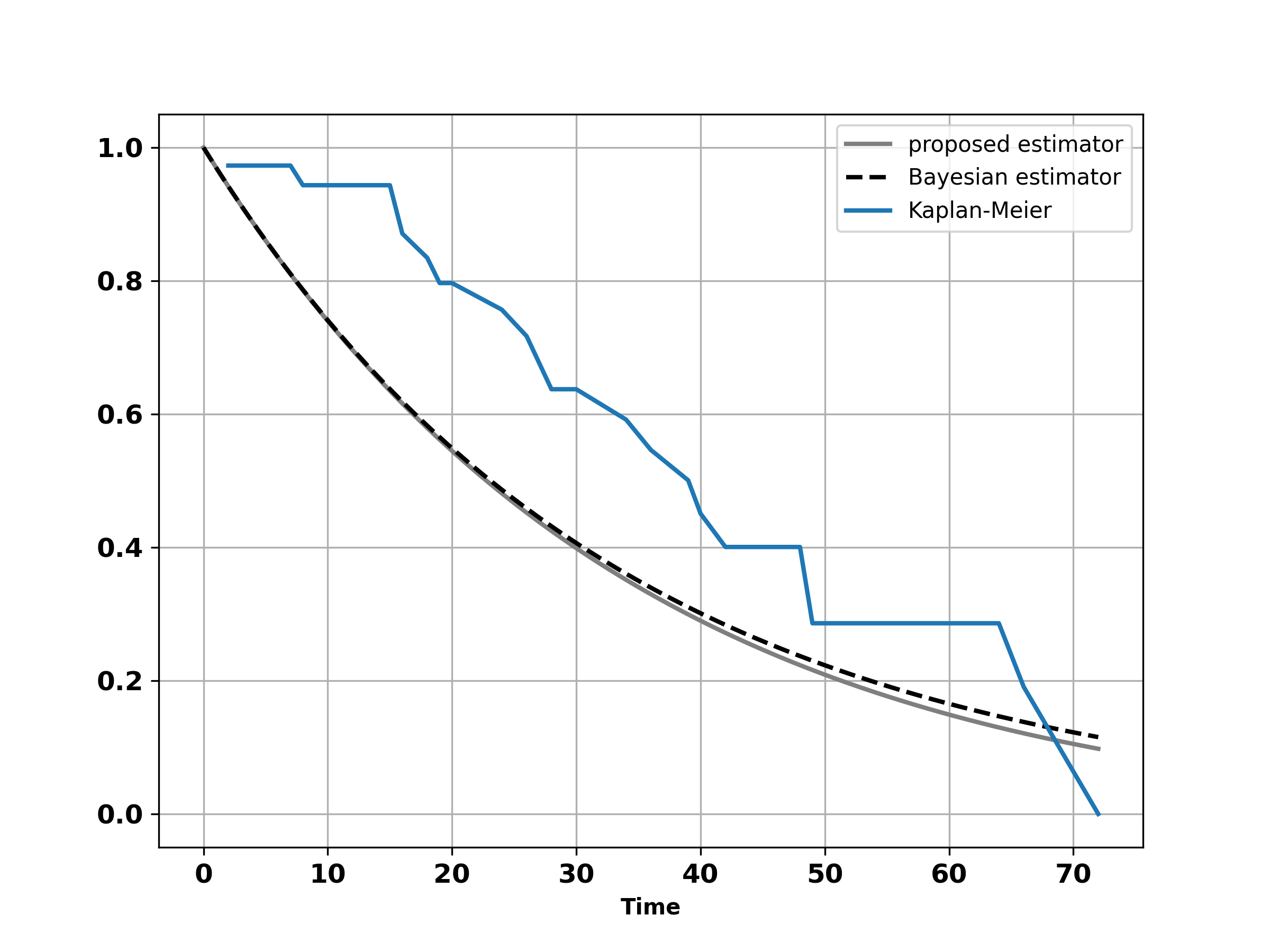} 
    \caption{Results for the survival function estimator of $T$ with UEFA Champion's League data.}
    \label{fig:foot}
\end{figure} 

\section{Conclusion remark}\label{sec:con}
\noindent 
The presence of dependent censoring remains a major obstacle in reliable survival analysis, often leading to biased inference when standard methods are used. In this work, we addressed this issue by introducing a modelling strategy based on the EGMO framework, combined with geometric optimization tools for parameter estimation. This approach allows a more faithful representation of the interaction between failure and censoring mechanisms.
Beyond the methodological construction, we provided theoretical guarantees for the proposed estimators, ensuring their asymptotic validity. The empirical results, supported by simulations and real-data applications, confirm that the method performs consistently well under a range of scenarios.

\section*{Funding}
This research did not receive any specific grant from funding agencies in the public, commercial, or not-for-profit sectors.

\section*{Declaration of interest}
The authors declare that they have no known competing financial interests or personal relationships that could have appeared to influence the work reported in this paper.

%\appendix
\section{Appendix}\label{sec:appendix}

In this section, we present the proofs of all the theorems.
\subsection*{\textbf{Appendix A. Proof of Lemma \ref{lemI}}}
By definition, $q_W^d \in \mathcal{Q}$ if and only if its $\mathbb{L}^2$-norm satisfies
$$
||q_W^d||_{\mathbb{L}^2}^2 = \int_I \big(q_W^d(t)\big)^2 \, dt = 1. 
$$
Let $\Phi(t) = (\phi_1(t), \dots, \phi_d(t))^T$ and $W = (w_1, \dots, w_d)^T$. Then we can write
$$
q_W^d(t) = W^T \Phi(t),
$$
so that
$$
||q_W^d||_{\mathbb{L}^2}^2 
= \int_I \big(W^T \Phi(t)\big)^2 \, \rho(t) dt
= \int_I W^T \Phi(t) \Phi(t)^T W \, \rho(t) dt
= W^T \Big( \int_I \Phi(t) \Phi(t)^T \, \rho(t) dt \Big) W.
$$
Since $(\phi_l)_{l=1}^d$ is an orthonormal basis in $\mathbb{L}^2(I)$, we have
\[
\int_I \Phi(t) \Phi(t)^T \, \rho(t) dt = I_d,
\]
where $I_d$ is the $d \times d$ identity matrix. Therefore,
\[
||q_W^d||_{\mathbb{L}^2}^2 = W^T I_d W = \sum_{l=1}^d w_l^2.
\]
Hence, $q_W^d \in \mathcal{Q}$ if and only if
\[
\sum_{l=1}^{d} w_l^2 = 1.
\]

\subsection*{\textbf{Appendix B. Proof of Lemma \ref{lemII}:}}
\noindent We must show that 
$$
\langle L_* \psi, L_* \phi \rangle_{\mathbb{L}^2} = \mathcal{G}(\psi, \phi)_f,
$$
where $L_*$ is the pushforward (or Jacobian) of $L$
\begin{equation*}
L_* = \frac{dL}{df}(f) = \frac{d(2\sqrt{f})}{df} = \frac{1}{\sqrt{f}}. \label{eq:pushforward}
\end{equation*}
\noindent By direct computation, we have
\begin{align*}
\langle L_* \psi, S_* \phi \rangle_{\mathbb{L}^2} 
&= \int_{\mathcal{D}} (L_* \psi)(t_1,t_2) \, (L_* \phi)(t_1,t_2) \, d t_1 d t_2 \\
&= \int_{\mathcal{D}} \psi(t_1,t_2) \, \frac{1}{\sqrt{f(t_1,t_2)}} \, \phi(t_1,t_2) \, \frac{1}{\sqrt{f(t_1,t_2)}} \, d t_1 d t_2 \\
&= \int_{\mathcal{D}} \psi(t_1,t_2) \, \phi(t_1,t_2) \, \frac{1}{f(t_1,t_2)} \, d t_1 d t_2 \\
&= \mathcal{G} (\psi, \phi)_f.
\end{align*}
\subsection*{\textbf{Appendix C. Proof of Lemma \ref{lrmIII}:}}
Since $(\xi_{l_1,l_2})_{l_1,l_2}$ is an orthonormal basis of $\mathbb{L}^2(\mathcal{D} , \rho)$, we have
$$
\langle \xi_{l_1,l_2}, \xi_{k_1,k_2} \rangle_{\mathbb{L}^2}
= \delta_{l_1,k_1}\delta_{l_2,k_2}
$$
For the truncated function
$
\psi_V^{m,p}(t_1,t_2)
=
\sum_{l_1=1}^m \sum_{l_2=1}^p v_{l_1,l_2}\,\xi_{l_1,l_2}(t_1,t_2),
$
its squared norm is given by
$$
\|\psi_V^{m,p}\|_{\mathbb{L}^2}^2
=
\left\langle
\sum_{l_1,l_2} v_{l_1,l_2}\xi_{l_1,l_2},
\sum_{k_1,k_2} v_{k_1,k_2}\xi_{k_1,k_2}
\right\rangle_{\mathbb{L}^2}
$$
By bilinearity of the inner product and orthonormality, this reduces to
$
\|\psi_V^{m,p}\|_{\mathbb{L}^2}^2
=
\sum_{l_1=1}^m \sum_{l_2=1}^p v_{l_1,l_2}^2
$.
Therefore, $\psi_V^{m,p} \in \mathcal{H}$ if and only if
$
\|\psi_V^{m,p}\|_{\mathbb{L}^2} = 1
$ which is equivalent to $
\sum_{l_1=1}^m \sum_{l_2=1}^p v_{l_1,l_2}^2 = 1$.
\subsection*{\textbf{Appendix D. Proof of Lemma \ref{lemIV}:}}
\noindent
	The geodesic on the sphere with initial data $ A_0 \in \mathcal{S}^{M-1} $ and $  \dot{A}_0 \in T_{A_0}\mathcal{S}^{M-1} $, with norm $ \alpha = || \dot{A}_0||_2 $, has the closed-form expression
	\begin{equation*}
		A(t) = A_0 \cos(\alpha t) + \frac{ \dot{A}_0}{\alpha} \sin(\alpha t).
	\end{equation*}
	Differentiating twice with respect to $ t $, we find
	\begin{equation*}
		\ddot{A}(t) = -\alpha^2 A(t).
	\end{equation*}
	The first derivative of $ F \circ A $ is
	\begin{equation*}
		\frac{d}{dt} F(A(t)) = \nabla F(A(t))^\top \dot{A}(t),
	\end{equation*}
	and the second derivative becomes
	\begin{equation*}
		\frac{d^2}{dt^2} F(A(t)) = \dot{A}(t)^\top \nabla^2 F(A(t)) \dot{A}(t) + \nabla F(A(t))^\top \ddot{A}(t).
	\end{equation*}
	Substituting $ \ddot{A}(t) = -\alpha^2 A(t) $, we obtain
	\begin{align*}
		\frac{d^2}{dt^2} F(A(t)) &= \dot{A}(t)^\top \nabla^2 F(A(t)) \dot{A}(t) - \alpha^2 \nabla F(A(t))^\top A(t) \\
		&= \dot{A}(t)^\top \left[ \nabla^2 F(A(t)) - \nabla F(A(t))^\top A(t) I_M \right] \dot{A}(t).
	\end{align*}
	Evaluating at $ t = 0 $, where $ A(0) = A_0 $ and $ \dot{A}(0) =  \dot{A}_0 $, concludes the proof.

\subsection*{\textbf{Appendix E. Proof of Proposition \ref{existance}:}}
\noindent
The proof of this theorem consists in verifying that the Assumptions \( A_1 \) and \( A_2 \) in \cite{Whi82} are fulfilled.\\
\underline{Proof of Assumption $A_1$} : 
The independent random vectors $(Y_i,\delta_i)$ for $i=1,\dots,n$ are identically distributed with a common measurable Radon-Nikodym density $g$ defined on the Euclidean space $\mathbb{R}^2$ equipped with the product $\sigma$-algebra $\mathcal{B}(\mathbb{R}^2)$ and product measure $\lambda \otimes \mu$, where $\lambda$ is the Lebesgue measure on $(\mathbb{R},\mathcal{B}(\mathbb{R}))$ and $\mu = \delta_0 + \delta_1$ is the sum of Dirac measures at $0$ and $1$.\\
\ \\
\underline{Proof of Assumption $A_2$} : 
 The function \( h \) is measurable with respect to the variables \( (t, \ell) \), for every parameter vector \( \theta  \) belonging to the compact subset $\Theta$ of the Euclidean space \( \mathbb{R}^{q} \) where $q=(m+1)(p+1)+(d+1)$, and continuous with respect to $\theta$, for every $(t,\ell) \in \mathbb{R}^2$. 
\subsection*{\textbf{Appendix F. Proof of Proposition \ref{consistance}:}}
\noindent
To prove this proposition, we verify Assumptions $A_1$–$A_3$ from \cite{Whi82}. Assumptions $A_1$ and $A_2$ follow directly from Proposition~\ref{existance}. It remains to verify assumption $A_3$, which we now address.\\
\ \\
 \underline{Proof of Assumption $A_3(a)$} : 
 We must show that the expectation \( \mathbb{E} \left[ \log g(Y, \delta,f_{1,2},f_3) \right] \) exists and that there exists an integrable function \( m \), with respect to the distribution of \( (Y, \delta) \), such that \( |\log h(y, \ell;\theta)| \leq m(t, \ell) \) for all $\theta \in \Theta$.\\
 \ \\
 $\star$ We start by proving that the expectation \( \mathbb{E} \left[ \log g(Y, \delta,f_{1,2},f_3) \right] \) exists. \\
 \ \\
 \underline{Case $\delta=0$} : We recall that for \( \delta = 0 \), the likelihood  is given by:
 \begin{eqnarray*}
  g(Y,0,f_{1,2},f_3)&=& f_{Y,\delta=0}(t)
 = \overline{F}_3(t) \int_t^{+\infty} f_{1,2}(u,t)du.
 \end{eqnarray*}
 Since $\overline{F}_3(Y) \leq 1$,  it follows that
 \begin{eqnarray*}
  \mathbb{E} \left[g(Y,0,f_{1,2},f_3) \right]&=& \mathbb{E} \left[ \log\left( \overline{F}_3(Y)  \int_Y^{+\infty} f_{1,2}(u,Y)du \right) \right]
  \\& \leq &
  \mathbb{E} \left[ \log\left(   \int_Y^{+\infty} f_{1,2}(u,Y)du \right) \right].
 \end{eqnarray*}
 Since $\log(x) \leq x$  for $x>0$, then 
 \begin{eqnarray}\label{eqI}
  \mathbb{E} \left[g(Y,0,f_{1,2},f_3) \right] \nonumber &\leq & \mathbb{E} \left[    \int_Y^{+\infty} f_{1,2}(u,Y)du  \right]
 \nonumber \\&=& \int_0^{+\infty} \int_t^{+\infty} f_{1,2}(u,t)du  f_Y(t)dt
\nonumber  \\&=& \int_0^{+\infty} \int_t^{+\infty} f_{1,2}(u,t) f_Y(t) du  dt.
  \end{eqnarray}
Based on Equation (\ref{survieY}), we obtain
\begin{eqnarray*}
    f_Y(t)&=& -\frac{\partial \overline{H}(t)}{\partial t}
    \\&=& - \overline{F}_3(t) \frac{\partial  \overline{F}_{1,2}(t,t)}{ \partial t}+ \overline{F}_{1,2}(t,t)f_3(t).
\end{eqnarray*}
Based on Leibniz integral rule, we have
\begin{eqnarray*}
 \frac{\partial  \overline{F}_{1,2}(t,t)}{ \partial t} &=& \frac{\partial}{\partial t} \left[ \int_t^{+\infty} \int_t^{+\infty} f_{1,2}(x,y)dxdy  \right] 
 \\&=& -\int_t^{+\infty} f_{1,2}(x,t)dx- \int_t^{+\infty} f_{1,2}(t,y)dy,
\end{eqnarray*}
which implies that,
\begin{eqnarray}\label{fY}
 f_Y(t)&=& \overline{F}_3(t)  \int_t^{+\infty} f_{1,2}(x,t)dx  + \overline{F}_3(t)\int_t^{+\infty} f_{1,2}(t,y)dy +  \overline{F}_{1,2}(t,t)f_3(t).
\end{eqnarray}
Substituting Equation (\ref{fY}) into Equation (\ref{eqI}), we obtain 
 \begin{eqnarray*}
   \mathbb{E} \left[g(Y,0,f_{1,2},f_3) \right] \nonumber &\leq &  I_1+I_2+I_3, 
 \end{eqnarray*}
 where 
 \begin{align*}
     I_1&= \int_0^{+\infty} \int_t^{+\infty} \int_t^{+\infty} f_{1,2}(u,t) f_{1,2} (x,t) \overline{F}_3(t) dxdudt,\\
     I_2&= \int_0^{+\infty} \int_t^{+\infty} \int_t^{+\infty} f_{1,2}(u,t) f_{1,2} (t,y) \overline{F}_3(t) dydudt,\\
     I_3&= \int_0^{+\infty} \int_t^{+\infty} f_{1,2}(u,t) \overline{F}_{1,2}(t,t)f_3(t) du  dt.
 \end{align*}
Now we verify that all the integrals above are finite. We have 
\begin{eqnarray*}
 I_1 &\leq & \int_0^{+\infty} \int_t^{+\infty} \int_t^{+\infty} f_{1,2}(u,t) f_{1,2} (x,t)  dxdudt
 \\&=&  \int_0^{+\infty} \int_t^{+\infty} f_{1,2} (x,t)  dx  \int_t^{+\infty}  f_{1,2}(u,t)dudt
 \\&=& \int_0^{+\infty} \left[  \int_t^{+\infty} f_{1,2} (x,t)  dx  \right]^2 dt.
\end{eqnarray*}
Using Hölder's inequality,
\begin{eqnarray*}
 I_2 &\leq &  \int_0^{+\infty} \int_t^{+\infty} f_{1,2} (t,y)  dy  \int_t^{+\infty}  f_{1,2}(u,t)dudt
 \\& \leq &  \left( \int_0^{+\infty} \left[ \int_t^{+\infty} f_{1,2} (t,y)  dy \right]^2 dt \right)^{1/2} \left( \int_0^{+\infty} \left[ \int_t^{+\infty}  f_{1,2}(u,t)du  \right]^2 dt \right)^{1/2},
\end{eqnarray*}
and 
\begin{eqnarray*}
 I_3 &\leq & \int_0^{+\infty} f_3(t) \int_t^{+\infty} f_{1,2}(u,t)  du  dt
 \\& \leq & \left( \int_0^{+\infty} f_3(t)^2dt  \right)^{1/2}  \left( \int_0^{+\infty} \left[ \int_t^{+\infty}  f_{1,2}(u,t)du  \right]^2 dt \right)^{1/2}.
\end{eqnarray*}
Based on Assumption $\mathbf{(H_1)}$, the integrals $I_1$, $I_2$ and $I_3$ are finite, which implies that $\mathbb{E}(\log g(Y,0,f_{1,2},f_3))$  exists.  
\\
\ \\
 \underline{Case $\delta=1$} :  We recall that for \( \delta = 1 \), the likelihood  is given by:
 \begin{eqnarray*}
  g(Y,1,f_{1,2},f_3)&=&f_Y(Y)-f_{Y,\delta=0}(Y).
 \end{eqnarray*}
 It follows that, 
 \begin{eqnarray*}
 \mathbb{E} \left[g(Y,0,f_{1,2},f_1) \right]\leq \mathbb{E} \left[f_Y(Y) \right]. 
 \end{eqnarray*}
 Based on Equation (\ref{fY}), we obtain
 \begin{eqnarray*}
  \mathbb{E} \left[g(Y,0,f_{1,2},f_1) \right]\leq  J_1+J_2+J_3, 
 \end{eqnarray*}
 where 
 \begin{align*}
     J_1&=\int_0^{+\infty} f_Y(t) \overline{F}_3(t)  \int_t^{+\infty} f_{1,2}(x,t)dx  dt ,\\
     J_2&= \int_0^{+\infty} f_Y(t) \overline{F}_3(t)\int_t^{+\infty} f_{1,2}(t,y)dy  dt,\\
     J_3&= \int_0^{+\infty} f_Y(t) \overline{F}_{1,2}(t,t)f_3(t) dt.
 \end{align*}
 We have,
 \begin{eqnarray*}
  J_1&=& \int_0^{+\infty} \left[ \overline{F}_3(t)  \int_t^{+\infty} f_{1,2}(x,t)dx  + \overline{F}_3(t)\int_t^{+\infty} f_{1,2}(t,y)dy +  \overline{F}_{1,2}(t,t)f_3(t) \right] \overline{F}_3(t)  \int_t^{+\infty} f_{1,2}(x,t)dx  dt 
  \\&=& \int_0^{+\infty} \overline{F}_3(t)^2 \left(\int_t^{+\infty} f_{1,2}(x,t)dx \right)^2 dt+\int_0^{+\infty} \overline{F}_3(t)^2 \left(\int_t^{+\infty} f_{1,2}(t,y)dy \right)  \left(\int_t^{+\infty} f_{1,2}(x,t)dx \right) dt
  \\&+& \int_0^{+\infty} \overline{F}_{1,2}(t,t) f_3(t) \left( \int_t^{+\infty} f_{1,2}(x,t)dx  \right) dt
  \\& \leq & \int_0^{+\infty}  \left(\int_t^{+\infty} f_{1,2}(x,t)dx \right)^2 dt+\int_0^{+\infty}  \left(\int_t^{+\infty} f_{1,2}(t,y)dy \right)  \left(\int_t^{+\infty} f_{1,2}(x,t)dx \right) dt
  \\&+& \int_0^{+\infty}  f_3(t) \left( \int_t^{+\infty} f_{1,2}(x,t)dx  \right) dt.
 \end{eqnarray*}
Using Hölder's inequality,
\begin{eqnarray*}
 J_1 & \leq & \int_0^{+\infty}  \left(\int_t^{+\infty} f_{1,2}(x,t)dx \right)^2 dt +\left( \int_0^{+\infty} \left[ \int_t^{+\infty} f_{1,2} (t,y)  dy \right]^2 dt \right)^{1/2} \left( \int_0^{+\infty} \left[ \int_t^{+\infty}  f_{1,2}(x,t)dx  \right]^2 dt \right)^{1/2}
 \\&+& \left( \int_0^{+\infty} f_3(t)^2dt  \right)^{1/2}  \left( \int_0^{+\infty} \left[ \int_t^{+\infty}  f_{1,2}(u,t)du  \right]^2 dt \right)^{1/2}. 
\end{eqnarray*}
Based on Assumption $\mathbf{(H_1)}$, the integral $J_1$ is finite. Similar reflections ensure that the integral $J_2$ is finite. Now for the integral $J_3$, we obtain,
\begin{eqnarray*}
 J_3 &=&   \int_0^{+\infty} \left[ \overline{F}_3(t)  \int_t^{+\infty} f_{1,2}(x,t)dx  + \overline{F}_3(t)\int_t^{+\infty} f_{1,2}(t,y)dy +  \overline{F}_{1,2}(t,t)f_3(t) \right] \overline{F}_{1,2}(t,t) f_3(t)  dt  \\& \leq & \int_0^{+\infty}  f_3(t) \left( \int_t^{+\infty} f_{1,2}(x,t)dx  \right) dt+ \int_0^{+\infty}  f_3(t) \left( \int_t^{+\infty} f_{1,2}(t,y)dy  \right) dt + \int_0^{+\infty} f_3(t)^2dt. 
\end{eqnarray*}
Using Hölder's inequality and based on Assumption $\mathbf{(H_1)}$, the integral $J_3$ is finite. Therefore, combining all the terms, we conclude that $\displaystyle{\mathbb{E} \left[\log g(Y,1,f_{1,2},f_3)\right] < \infty}$.\\
\ \\
$\star \star$ We now turn to establishing a bound on $|\log h(y, \ell;\theta)|$ by an integrable function with respect to $g$, for all $\theta \in \Theta$.  \\
 \ \\
\underline{Case $\ell=0$}: When $\ell=0$, we have 
$$\left| \log \left(h(t,\ell;\theta) \right) \right|=\left| \log \left( \widetilde{f}_{Y,\delta=0}(t) \right) \right|.$$
\vspace{0.2cm}
\noindent\underline{Subcase 1}: 
 If $h(t,\ell; \theta)<1$, then
$$\left|\log(h(t,\ell;\theta))\right|=-\log(h(t,\ell;\theta).$$
Define the following bounds of parameters space:
\[
V_{k,\min} = \min \{ V_k \in \Theta_1 \}, \quad V_{k,\max} = \max \{ V_k \in \Theta_1 \}, \quad \text{for } k \in \{0, \dots, mp\},
\]

\[
W_{k,\min} = \min \{ W_k \in \Theta_2 \}, \quad W_{k,\max} = \max \{ W_k \in \Theta_2 \}, \quad \text{for } k \in \{0, \dots, d \}
\]
For each $k \in \{1\dots mp\}$ and for each $j \in \{1\dots d\}$, we have 
$$ V_{k,min} < V_k < V_{k,max} \qquad \text{and} \qquad W_{j,\min} < V_j < V_{j,\max},$$
which  implies the following bounds:
\begin{eqnarray*}
 \overline{F}_{3,d}(t|\overline{V}_{d,\min}) < \overline{F}_{3,d}(t|\overline{V}_d) <  \overline{F}_{3,d}(t|\overline{V}_{d,\max}), 
\end{eqnarray*}
and 
\begin{eqnarray*}
 f_{1,2,mp}(t,s|\overline{W}_{mp,\min}) < f_{1,2,mp}(t,s|\overline{W}_{mp}) <  \overline{F}_{1,2,mp}(t,s|\overline{W}_{mp,\max}).
\end{eqnarray*}
Then 
\begin{eqnarray*}
 \int_t^b f_{1,2,mp} (u,t|\overline{W}_{mp,\min})du<  \int_t^b f_{1,2,mp} (u,t|\overline{W}_{mp})du < \int_t^b f_{1,2,mp} (u,t|\overline{W}_{mp,\max})du.
\end{eqnarray*}
Then 
\begin{eqnarray}\label{eqmaj0}
 \overline{F}_{3,d}(t|\overline{V}_{d,\min}) \int_t^{b} f_{1,2,mp}(u,t|\overline{W}_{mp,\min})du & \leq &
 \overline{F}_{3,d}(t|\overline{V}_{d}) \int_t^{b} f_{1,2,mp}(u,t|\overline{W}_{m,p})du  
\nonumber\\& \leq & \overline{F}_{3,d}(t|\overline{V}_{d,\max}) \int_t^{b} f_{1,2,mp}(u,t|\overline{W}_{mp,\max})du.
\end{eqnarray}
Since, we have
\begin{eqnarray*}
\log(h(t,\ell;\theta))=\log\left(\widetilde{f}_{Y,\delta=0}(t) \right) &=& \log \left( \overline{F}_{3,d}(t|\overline{V}_{d}) \int_t^{b} f_{1,2,mp}(u,t|\overline{W}_{mp})du \right),
\end{eqnarray*}
we obtain the following inequality
\begin{eqnarray*}
 \log \left( \overline{F}_{3,d}(t|\overline{V}_{d,\min}) \int_t^{b} f_{1,2,mp}(u,t|\overline{W}_{mp,\min})du \right) & \leq & \log(h(t,\ell;\theta)) 
 \\& \leq &  \log \left( \overline{F}_{3,d}(t|\overline{V}_{d,\max}) \int_t^{b} f_{1,2,mp}(u,t|\overline{W}_{mp,\max})du \right)
\end{eqnarray*}
Then, 
\begin{eqnarray*}
-  \log \left( \overline{F}_{3,d}(t|\overline{V}_{d,\min}) \int_t^{b} f_{1,2,mp}(u,t|\overline{W}_{mp,\min})du \right) & \geq & - \log(h(t,\ell;\theta)) 
 \\& \geq &  - \log \left( \overline{F}_{3,d}(t|\overline{V}_{d,\max}) \int_t^{b} f_{1,2,mp}(u,t|\overline{W}_{mp,\max})du \right)
\end{eqnarray*}
Hence, there exists an integrable function 
$$m(t):=-  \log \left( \overline{F}_{3,d}(t|\overline{V}_{d,\min}) \int_t^{b} f_{1,2,mp}(u,t|\overline{W}_{mp,\min})du \right)$$ such that:
$$|\log (h(t,\ell;\theta)) | \leq m(t).$$
\vspace{0.2cm}
\\
\noindent\underline{Subcase 2}: If $h(t,\ell; \theta)>1$, then
$$\left|\log(h(t,\ell;\theta))\right|=\log(h(t,\ell;\theta).$$
\begin{eqnarray*}
\log(h(t,\ell;\theta))=\log\left(\widetilde{f}_{Y,\delta=0}(t) \right) &=& \log \left( \overline{F}_{3,d}(t|\overline{V}_{d}) \int_t^{b} f_{1,2,mp}(u,t|\overline{W}_{mp})du \right)
\\& \leq & \log \left(  \int_t^{b} f_{1,2,mp}(u,t|\overline{W}_{mp})du \right)
\\ &\leq &   \int_t^{b} \left[\sum_{i=1}^m \sum_{j=1}^p v_{i,j} \xi_{i,j}(u,t) \right]^2du
\\ &\leq &   \int_t^{b} \left[\sum_{i=1}^{+\infty} \sum_{j=1}^{+\infty} v_{i,j} \xi_{i,j}(u,t) \right]^2du
\\ &=&   \int_t^{b} \left[f_{1,2}(u,t) \right]^2du < \infty,
\end{eqnarray*}
since $f_{1,2} \in \mathbb{L}^2(D)$. 
\\
\ \\
\underline{Case $\ell=1$}:  In this case, we have
\begin{eqnarray*}
 h(t,\ell; \theta)&=& \widetilde{f}_{Y,n}(t)-\widetilde{f}_{Y,\delta=0}^{m,p,d}(t) 
 \\&=&  \widetilde{f}_{Y,n}(t)- \overline{F}_{3,W}^d(t) \int_t^{b} f_{1,2,V}^{m,p}(u,t)du.
\end{eqnarray*}
Using the bounds established in Equation (\ref{eqmaj0}), we deduce that

\begin{eqnarray*}
 \lefteqn{m(t):=\widetilde{f}_{Y,n}(t)-\overline{F}_{3,d,\max}(t|\overline{V}) \int_t^{b} f_{1,2,m,p,\max}(u,t|\overline{W})du}
  \\& \leq &
  h(t,\ell;\theta)
 \nonumber\\& \leq & 
 \widetilde{f}_{Y,n}(t)- \overline{F}_{3,d,\min}(t|\overline{V}) \int_t^{b} f_{1,2,m,p,\min}(u,t|\overline{W})du:=n(t).
\end{eqnarray*}
Taking the logarithm, we obtain:
It follows that
\[
\left| \log(h(t, \ell; \theta)) \right| \leq \max \left\{ \left| \log(m(t)) \right|, \left| \log(n(t)) \right| \right\}.
\]
Therefore, \( \left| \log(h(t, \ell; \theta)) \right| \) is uniformly bounded by an integrable function  for all \( \theta \). This completes the verification of condition \( A_3(a) \).\\
 \ \\
 \underline{Proof of Assumption $A_3(b)$} : 
 We need to show that $\text{KL}$ has unique minimum at $\theta^{\star}$.
The objective of this proof is to show that for $m,p,d \to \infty$, $\text{KL}(g||h;\theta)$ has unique minimum at $\theta^{\star}$, 
where 
First, we note that $$\text{KL}(g||h;\theta)=\mathbb{E} \left( \log \left[ \frac{g(Y,\delta;f_{1,2},f_3)}{h(Y,\delta;\theta) } \right]\right) \geq 0.$$  
Then the $\displaystyle{\text{KL}(g \| h; \theta)}$ divergence is minimized when  $f$ exactly matches $g$. Let $$\theta^{\star}:=(\theta_1^{\star},\dots,\theta_q^{\star})=\left((\mathcal{R}_{i,j}^{\star})_{1 \leq i \leq m, 1 \leq j \leq p}, (\mathcal{T}_{k}^{\star})_{1 \leq k \leq d} \right),$$
where $$\mathcal{T}_{k}^{\star} =
\big<\overline{F}_3,\phi \big>_{\mathbb{L}^2}=\int_{I} \overline{F}_3(t) \phi(t) \mathrm{d}t,$$
and 
$$\mathcal{R}_{i,j}^{\star} = \big<f_{1,2},\xi \big>_{\mathbb{L}^2}=\int_{D} f_{1,2}(t_1,t_2) \xi(t_1,t_2) \mathrm{d}t_1 \mathrm{d}t_2$$
such that $$ \sum_{i=1}^d (\mathcal{T}_i^{\star})^2=1 \quad \text{and} \quad  \sum_{i=1}^m \sum_{j=1}^p (\mathcal{R}_{i,j}^{\star})^2=1.$$
We have 
\begin{eqnarray}\label{lim1}
 \lim\limits_{m,p \to \infty} \sum_{i=0}^m \sum_{j=0}^p  \mathcal{R}_{i,j}^{\star} \xi_{i,j}(x,y)=f_{1,2} (x,y), \quad \quad \lim\limits_{d \to \infty} \sum_{k=0}^d  \mathcal{T}_{k}^{\star} \phi_k(x)=\overline{F}_3(x),   
\end{eqnarray}
and then 
$$\lim\limits_{m,p,d \to \infty} h(Y,\delta;\theta^{\star})=g(Y,\delta;f_{1,2},f_3) .$$ Then $\displaystyle{\text{KL}(g \| h; \theta)}$ divergence is minimized at $\theta^{\star}$. The uniqueness of $\theta^{\star}$  is guaranteed by the uniqueness of the projection in $\mathbb{L}^2$.
\subsection*{\textbf{Appendix G. Proof of Proposition \ref{aymnormal}}}
To establish this proposition, we verify Assumptions A1–A6 of \cite{Whi82}. Assumptions A1–A3 follow directly from Proposition 2. It remains to check Assumptions A4 and A5 of \cite{Whi82}.\\
\ \\
\underline{Proof of Assumption $A_4$} : 
We first express the log-likelihood as:
\begin{eqnarray*}
\log(h(t,\ell;\theta))&=& \ell \log \left[ \widetilde{f}_{Y,n}(t)-\widetilde{f}_{Y,\delta=0}^{m,p,d}(t)\right]+(1-\ell)\log \left[\widetilde{f}_{Y,\delta=0}^{m,p,d}(t)  \right]
\\&=&  \ell \log \left[ \widetilde{f}_{Y,n}(t)- \overline{F}_{3,W}^d(t) \int_t^{b} f_{1,2,V}^{m,p}(u,t)du\right]+(1-\ell)\log \left[  \overline{F}_{3,W}^d(t) \int_t^{b} f_{1,2,V}^{m,p}(u,t)du  \right].
\end{eqnarray*}
Next, we compute the partial derivatives of $\log(h(t,\ell;\theta))$ with respect to each parameter. Let 
$$\displaystyle{B:=\overline{F}_{3,W}^d(t) \int_t^{b} f_{1,2,V}^{m,p}(u,t)du}.$$ Derivative
with respect to $w_k$ for  $k=1\dots d$:
\begin{eqnarray*}
 \frac{\partial \log(h(t,\ell;\theta))  }{\partial w_k} &=&  
 \ell \dfrac{- \frac{\partial B}{\partial w_k} }{B}+ (1-\ell) \dfrac{ \frac{\partial B}{\partial w_k} }{B} 
 \\&=& \ell \dfrac{-  2 \int_t^{m,p} (u,y)du M(y) w_k}{B}+ (1-\ell) \dfrac{  2 \int_t^{m,p} (u,y)du M(y) w_k }{B} 
 \\&=& \dfrac{ 2 \int_t^b f_{1,2,V}^{m,p} (u,y)du M(y) w_k (1-2\ell)}{\overline{F}_{3,W}^d(t) \int_t^{b} f_{1,2,V}^{m,p}(u,t)du}
 \\&=& \dfrac{ 2 M(y) w_k (1-2\ell)}{\overline{F}_{3,W}^d(t)}.
\end{eqnarray*}
Derivative with respect to $v_{i,j}$ for  $i=1 \dots m$  and $j=1 \dots p$
\begin{eqnarray*}
  \frac{\partial \log(h(t,\ell;\theta))  }{\partial v_{i,j}} &=&  \ell \dfrac{- \frac{\partial B}{\partial v_{i,j}} }{B}+ (1-\ell) \dfrac{ \frac{\partial B}{\partial v_{i,j}} }{B} 
  \\&=& \ell \dfrac{- \frac{\partial B}{\partial v_{i,j}} }{B}+ (1-\ell) \dfrac{ \frac{\partial B}{\partial v_{i,j}} }{B} 
\end{eqnarray*}
Or based on Equation (\ref{partialBV}), we have  
\begin{eqnarray*}
 \frac{\partial \log(h(t,\ell;\theta))  }{\partial v_{i,j}} &=&  \ell \dfrac{- 	\overline{F}_{3,W}^d(t)
	\times 2 \int_{t}^{b} \left( V^T \boldsymbol{\xi}(u, t) \right) \times
	\boldsymbol{\xi}(u, t) \, du }{\overline{F}_{3,W}^d(t) \int_t^{b} f_{1,2,V}^{m,p}(u,t)du}+ (1-\ell) \dfrac{	\overline{F}_{3,W}^d(t)
	\times 2 \int_{t}^{b} \left( V^T \boldsymbol{\xi}(u, t) \right) \times
	\boldsymbol{\xi}(u, y) \, du }{\overline{F}_{3,W}^d(t) \int_t^{b} f_{1,2,V}^{m,p}(u,t)du} 
 \\&=& \ell \dfrac{- 	
 2 \int_{t}^{b} \left( V^T \boldsymbol{\xi}(u, t) \right) \times
	\boldsymbol{\xi}(u, t) \, du }{ \int_t^{b} f_{1,2,V}^{m,p}(u,t)du}+ (1-\ell) \dfrac{
 2 \int_{t}^{b} \left( V^T \boldsymbol{\xi}(u, t) \right) \times
	\boldsymbol{\xi}(u, t) \, du }{ \int_t^{b} f_{1,2,V}^{m,p}(u,t)du} 
	\\&=&   \dfrac{2 \int_{t}^{b} \left( V^T \boldsymbol{\xi}(u, t) \right) \times
	\boldsymbol{\xi}(u, t) \, du (1-2\ell)}{ \int_t^{b} f_{1,2,V}^{m,p}(u,t)du}
\end{eqnarray*}
The functions \( \frac{\partial \log(h(t,\ell;\theta))  }{\partial w_k} \) for \( k = 0, \dots, m \) and \( \frac{\partial \log(h(t,\ell;\theta))  }{\partial v_{i,j}}\), for  $i=1 \dots m$  and $j=1 \dots p$, are measurable functions of \( (t,\ell) \) for each fixed \( \theta \in \Theta_1 \times \Theta_2\), and are continuously differentiable functions of \( \theta \) for every \( (t,\ell) \in \, ]0, +\infty[ \times \mathbb{R} \). Therefore, Assumption \( A_4 \) is satisfied.\\
\ \\
\underline{Proof of Assumption $A_5$:  }  
We first consider second-order derivatives with respect to the parameter $V$.
The first-order derivative with respect to the entry $v_{i,j}$ of the matrix $V$
is given by
\[
\frac{\partial}{\partial v_{i,j}}\log h(t,\ell;\theta)
=
\frac{2(1-2\ell)}{\int_t^b f_{1,2,V}^{m,p}(u,t)\,du}
\int_t^b \big(V^\top\boldsymbol{\xi}(u,t)\big)\,
\xi_{i,j}(u,t)\,du.
\]
Differentiating again with respect to $v_{k,r}$ yields
\[
\frac{\partial^2}{\partial v_{i,j}\partial v_{k,r}}
\log h(t,\ell;\theta)
=
2(1-2\ell)
\left[
\frac{
\int_t^b \xi_{k,r}(u,t)\,\xi_{i,j}(u,t)\,du
}{
\int_t^b f_{1,2,V}^{m,p}(u,t)\,du
}
-
\frac{
\int_t^b (V^\top\boldsymbol{\xi}(u,t))\xi_{i,j}(u,t)\,du
\;
\int_t^b 2(V^\top\boldsymbol{\xi}(u,t))\xi_{k,r}(u,t)\,du
}{
\left(\int_t^b f_{1,2,V}^{m,p}(u,t)\,du\right)^2
}
\right].
\]
We now bound this expression. Since the basis functions $\boldsymbol{\xi}(u,t)$
are assumed to be bounded on the compact domain $[a,b]^2$, the products
$\xi_{i,j}(u,t)\xi_{k,r}(u,t)$ and
$(V^\top\boldsymbol{\xi}(u,t))\xi_{i,j}(u,t)$
are bounded functions of $(u,t)$.
Because the integration interval $[t,b]$ is finite, the integrals appearing in
the numerator are bounded functions of $t$.
Moreover, by Assumption $A_4$, the quantity
\[
\int_t^b f_{1,2,V}^{m,p}(u,t)\,du
\]
is uniformly bounded away from zero for all admissible $V$ and all $t$.
Consequently, the ratio defining the second derivative is bounded by a function
of $t$ that depends only on suprema of bounded integrands over finite intervals.
We denote this dominating function by
\[
g_V(t)
=
2\left|
\frac{
\int_t^b \xi_{k,r}(u,t)\,\xi_{i,j}(u,t)\,du
}{
\int_t^b f_{1,2,V}^{m,p}(u,t)\,du
}
\right|
+
2\left|
\frac{
\int_t^b (V^\top\boldsymbol{\xi}(u,t))\xi_{i,j}(u,t)\,du
\;
\int_t^b 2(V^\top\boldsymbol{\xi}(u,t))\xi_{k,r}(u,t)\,du
}{
\left(\int_t^b f_{1,2,V}^{m,p}(u,t)\,du\right)^2
}
\right|.
\]
Observe that the functions $\xi_{i,j}(u,t)$ are bounded on the compact domain
$[a,b]^2$. Since the parameter matrix $V$ belongs to a compact set, the quantity
$V^\top \boldsymbol{\xi}(u,t)$ is uniformly bounded in $(u,t)$. It follows that the products $\xi_{i,j}(u,t)\xi_{k,r}(u,t)$ and $(V^\top\boldsymbol{\xi}(u,t))\xi_{i,j}(u,t)$ are bounded functions. Because the integrals are taken over the finite interval $[t,b]$, the resulting
integrals are bounded functions of $t$. Moreover, Assumption~A4 ensures that
$\int_t^b f_{1,2,V}^{m,p}(u,t)\,du$ is finite and uniformly bounded away from zero.
Hence, $g_V(t)$ is bounded on the support of $t$, which implies that it is
integrable with respect to $G$.\\
\ \\
We next consider second-order derivatives with respect to $W$.
The first-order derivative with respect to $w_k$ is
\[
\frac{\partial}{\partial w_k}\log h(t,\ell;\theta)
=
\frac{2(1-2\ell) M(t)w_k}{\overline F_{3,W}^d(t)},
\qquad
\overline F_{3,W}^d(t)=W^\top M(t)W.
\]
Differentiating with respect to $w_m$ gives
\[
\frac{\partial^2}{\partial w_k\partial w_m}
\log h(t,\ell;\theta)
=
2(1-2\ell)
\left[
\frac{M(t)\delta_{k,m}}{\overline F_{3,W}^d(t)}
-
\frac{M(t)w_k \cdot 2(M(t)W)_m}{(\overline F_{3,W}^d(t))^2}
\right].
\]
Since $\Phi$ is bounded, the matrix-valued function
$M(t)=\int_t^b \Phi(u)\Phi(u)^\top du$ is bounded for all $t$.
Moreover, Assumption~A4 ensures that $\overline F_{3,W}^d(t)$ is uniformly bounded
away from zero.
Therefore, the above expression is bounded in absolute value by
\[
g_W(t)
=
2\left|
\frac{M(t)}{\overline F_{3,W}^d(t)}
\right|
+
2\left|
\frac{M(t)M(t)W}{(\overline F_{3,W}^d(t))^2}
\right|.
\]
Since the function $\Phi$ is bounded on $[a,b]$, the matrix $M(t)$ is bounded for all $t$. The compactness of the parameter space for $W$ implies that both $W$ and $M(t)W$ are bounded vectors. By Assumption $A_4$, the quantity
$\overline F_{3,W}^d(t)$ is finite and bounded away from zero. Therefore, each term in the expression of $g_W(t)$ is bounded on the support of
$t$. This implies that $g_W(t)$ is integrable with respect to $G$.\\
\ \\
Finally, we consider the mixed second-order derivatives.
Using the expression of $B(t;\theta)$,
\[
B(t;\theta)
=
\overline F_{3,W}^d(t)\int_t^b f_{1,2,V}^{m,p}(u,t)\,du,
\]
we obtain
\[
\frac{\partial^2}{\partial v_{i,j}\partial w_k}
\log h(t,\ell;\theta)
=
-4(1-2\ell)
\frac{
\int_t^b (V^\top\boldsymbol{\xi}(u,t))\xi_{i,j}(u,t)\,du
\cdot (M(t)W)_k
}{
\overline F_{3,W}^d(t)
\int_t^b f_{1,2,V}^{m,p}(u,t)\,du
}.
\]
Each term in the numerator is bounded as argued above, while the denominator
is uniformly bounded away from zero by Assumption $A_4$. Hence,
\[
\left|
\frac{\partial^2}{\partial v_{i,j}\partial w_k}
\log h(t,\ell;\theta)
\right|
\le
g_{VW}(t),
\]
where
\[
g_{VW}(t)
=
4\left|
\frac{
\int_t^b |(V^\top\boldsymbol{\xi}(u,t))\xi_{i,j}(u,t)|\,du
\cdot |(M(t)W)_k|
}{
\overline F_{3,W}^d(t)
\int_t^b f_{1,2,V}^{m,p}(u,t)\,du
}
\right|.
\]
Using the previous bounds, the integral
$\int_t^b (V^\top\boldsymbol{\xi}(u,t))\xi_{i,j}(u,t)\,du$ is a bounded function of
$t$. In addition, the boundedness of $M(t)$ and the compactness of the parameter space for $W$ imply that $(M(t)W)_k$ is bounded.
Together with Assumption $A_4$, which guarantees that the denominator is finite and uniformly bounded away from zero, we conclude that $g_{VW}(t)$ is bounded on the support of $t$. Consequently, $g_{VW}(t)$ is integrable with respect to $G$.\\
\ \\
We now consider the products of the first-order derivatives with respect to the
parameters of $V$. Using the expression obtained previously, we have
\[
\frac{\partial}{\partial v_{i,j}}\log h(t,\ell;\theta)
=
\frac{2(1-2\ell)}{\int_t^b f_{1,2,V}^{m,p}(u,t)\,du}
\int_t^b (V^\top\boldsymbol{\xi}(u,t))\xi_{i,j}(u,t)\,du.
\]
Therefore,
\[
\left|
\frac{\partial}{\partial v_{i,j}}\log h(t,\ell;\theta)
\frac{\partial}{\partial v_{k,r}}\log h(t,\ell;\theta)
\right|
=
\frac{
4\left|
\int_t^b (V^\top\boldsymbol{\xi}(u,t))\xi_{i,j}(u,t)\,du
\int_t^b (V^\top\boldsymbol{\xi}(u,t))\xi_{k,r}(u,t)\,du
\right|
}{
\left(\int_t^b f_{1,2,V}^{m,p}(u,t)\,du\right)^2
}.
\]
As shown above, the integrals in the numerator are bounded functions of $t$,
while the denominator is finite and uniformly bounded away from zero by
Assumption $A_4$. Hence, the above product is bounded by a function of $t$.
We denote this dominating function by
\[
\tilde g_V(t)
=
\frac{
4\left|
\int_t^b (V^\top\boldsymbol{\xi}(u,t))\xi_{i,j}(u,t)\,du
\int_t^b (V^\top\boldsymbol{\xi}(u,t))\xi_{k,r}(u,t)\,du
\right|
}{
\left(\int_t^b f_{1,2,V}^{m,p}(u,t)\,du\right)^2
}.
\]
Using the same boundedness arguments as before, it follows that $\tilde g_V(t)$ is bounded on the support of $t$ and thus integrable with respect to $G$.\\
\ \\
We next consider the products of the first-order derivatives with respect to the
parameters of $W$. Recall that
\[
\frac{\partial}{\partial w_k}\log h(t,\ell;\theta)
=
\frac{2(1-2\ell) M(t)w_k}{\overline F_{3,W}^d(t)}.
\]
Hence,
\[
\left|
\frac{\partial}{\partial w_k}\log h(t,\ell;\theta)
\frac{\partial}{\partial w_m}\log h(t,\ell;\theta)
\right|
=
\frac{
4(1-2\ell)^2 |M(t)w_k M(t)w_m|
}{
(\overline F_{3,W}^d(t))^2
}.
\]
Since $M(t)$ and $W$ are bounded and $\overline F_{3,W}^d(t)$ is finite and
uniformly bounded away from zero by Assumption $A_4$, this product is bounded by a
function of $t$. We define the dominating function as
\[
\tilde g_W(t)
=
\frac{
4|M(t)W|^2
}{
(\overline F_{3,W}^d(t))^2
}.
\]
As previously argued, $\tilde g_W(t)$ is bounded on the support of $t$ and hence integrable with respect to $G$.\\
\ \\
Finally, we consider the products of first-order derivatives involving both $V$
and $W$. We obtain
\[
\left|
\frac{\partial}{\partial v_{i,j}}\log h(t,\ell;\theta)
\frac{\partial}{\partial w_k}\log h(t,\ell;\theta)
\right|
=
\frac{
4\left|
\int_t^b (V^\top\boldsymbol{\xi}(u,t))\xi_{i,j}(u,t)\,du
\cdot M(t)w_k
\right|
}{
\overline F_{3,W}^d(t)
\int_t^b f_{1,2,V}^{m,p}(u,t)\,du
}.
\]
The numerator is bounded as previously shown, while the denominator is finite
and uniformly bounded away from zero by Assumption $A_4$. Hence, this product is
bounded by a function of $t$, denoted by
\[
\tilde g_{VW}(t)
=
\frac{
4\left|
\int_t^b (V^\top\boldsymbol{\xi}(u,t))\xi_{i,j}(u,t)\,du
\cdot M(t)w_k
\right|
}{
\overline F_{3,W}^d(t)
\int_t^b f_{1,2,V}^{m,p}(u,t)\,du
}.
\]
Using the same arguments as above, $\tilde g_{VW}(t)$ is bounded and integrable with respect to $G$, which completes the verification of Assumption $A_5$.

\bibliography{bibliography}

@article{esc24,
  title={Dependent censoring with simultaneous death times based on the Generalized Marshall--Olkin model},
  author={Escobar-Bach, Mikael and Helali, Salima},
  journal={Journal of Multivariate Analysis},
  volume={204},
  pages={105347},
  year={2024},
  publisher={Elsevier}
}

@article{pin15,
  title={Extended Marshall--Olkin model and its dual version},
  author={Pinto, Jayme and Kolev, Nikolai},
  booktitle={Marshall ̶ Olkin Distributions-Advances in Theory and Applications: Bologna, Italy, October 2013},
  pages={87--113},
  year={2015},
  publisher={Springer}
}

@article{Whi82,
  author = {White, Halbert},
  year = {1982},
  title = {Maximum likelihood estimation of misspecified models},
  journal = {Econometrica},
  volume = {50},
  pages = {1--25}
}

@article{Tsia75,
  author = {Tsiatis, Anastasios A},
  year = {1975},
  title = {A nonidentifiability aspect of the problem of competing risks},
  journal = {Proc. Natl. Acad. Sci. USA},
  volume = {72},
  pages = {20--22}
}

@book{Kal02,
  title={The statistical analysis of failure time data},
  author={Kalbfleisch, John D and Prentice, Ross L},
  year={2002},
  publisher={John Wiley \& Sons}
}

@book{kle03,
  title={Survival analysis: techniques for censored and truncated data},
  author={Klein, John P and Moeschberger, Melvin L},
  volume={1230},
  year={2003},
  publisher={Springer}
}

@book{Fle13,
  title={Counting processes and survival analysis},
  author={Fleming, Thomas R and Harrington, David P},
  year={2013},
  publisher={John Wiley \& Sons}
}

@article{Tur76,
  title={The empirical distribution function with arbitrarily grouped, censored and truncated data},
  author={Turnbull, Bruce W},
  journal={Journal of the Royal Statistical Society: Series B (Methodological)},
  volume={38},
  number={3},
  pages={290--295},
  year={1976},
  publisher={Wiley Online Library}
}

@article{Kap58,
  title={Nonparametric estimation from incomplete observations},
  author={Kaplan, Edward L and Meier, Paul},
  journal={Journal of the American statistical association},
  volume={53},
  number={282},
  pages={457--481},
  year={1958},
  publisher={Taylor \& Francis}
}

@article{Cox72,
  title={Regression models and life-tables},
  author={Cox, David R},
  journal={Journal of the royal statistical society: Series B (methodological)},
  volume={34},
  number={2},
  pages={187--202},
  year={1972},
  publisher={Wiley Online Library}
}

@inproceedings{Gil97,
  title={Coarsening at random: Characterizations, conjectures, counter-examples},
  author={Gill, Richard D and Van Der Laan, Mark J and Robins, James M},
  booktitle={Proceedings of the First Seattle Symposium in Biostatistics: Survival Analysis},
  pages={255--294},
  year={1997},
  organization={Springer}
}

@book{Tsi06,
  title={Semiparametric theory and missing data},
  author={Tsiatis, Anastasios A},
  year={2006},
  publisher={Springer}
}

@article{Ibr13,
  title={Bayesian analysis of the Cox model},
  author={Ibrahim, Joseph G and Chen, Ming-Hui and Zhang, Danjie and Sinha, Debajyoti},
  journal={Handbook of survival analysis},
  volume={27},
  year={2013},
  publisher={Chapman \& Hall/CRC Handbooks of Modern Statistical Methods. Chapman \& Hall~…}
}

@article{Tsa85,
  title={A large sample study of generalized maximum likelihood estimators from incomplete data via self-consistency},
  author={Tsai, Wei-Yann and Crowley, John},
  journal={The Annals of Statistics},
  pages={1317--1334},
  year={1985},
  publisher={JSTOR}
}

@article{Rob94,
  title={Estimation of regression coefficients when some regressors are not always observed},
  author={Robins, James M and Rotnitzky, Andrea and Zhao, Lue Ping},
  journal={Journal of the American statistical Association},
  volume={89},
  number={427},
  pages={846--866},
  year={1994},
  publisher={Taylor \& Francis}
}

@book{Laa03,
  title={Unified methods for censored longitudinal data and causality},
  author={Laan, Mark J and Robins, James M},
  year={2003},
  publisher={Springer}
}

@article{Mar05,
  title={A new method for adding a parameter to a family of distributions with application to the exponential and weibull families},
  author={Marshall, Albert W and Olkin, Ingram},
  journal={Biometrika},
  volume={92},
  number={2},
  pages={505--505},
  year={2005},
  publisher={Biometrika Trust}
}

@article{Dia16,
  title={Exponentiated Marshall-Olkin family of distributions},
  author={B. Dias, C{\'\i}cero R and Cordeiro, Gauss M and Alizadeh, Morad and Diniz Marinho, Pedro Rafael and Campos Co{\^e}lho, Hem{\'\i}lio Fernandes},
  journal={Journal of Statistical Distributions and Applications},
  volume={3},
  number={1},
  pages={15},
  year={2016},
  publisher={Springer}
}

@misc{Mar00,
  title={Directional statistics. Chichester: John willey and sons},
  author={Mardia, K and Jupp, P},
  year={2000},
  publisher={Inc}
}

@InProceedings{pmlr-v124-holbrook20a,
  title = 	 {Nonparametric {F}isher geometry with application to density estimation},
  author =       {Holbrook, Andrew and Lan, Shiwei and Streets, Jeffrey and Shahbaba, Babak},
  booktitle = 	 {Proceedings of the 36th Conference on Uncertainty in Artificial Intelligence (UAI)},
  pages = 	 {101--110},
  year = 	 {2020},
  volume = 	 {124},
  series = 	 {Proceedings of Machine Learning Research},
  publisher =    {PMLR}
}

@article{Fradi2024,
  author    = {Fradi, Anis and Samir, Chafik and Adouani, Ines},
  title     = {A new {B}ayesian approach to global optimization on parametrized Surfaces in {$\mathbb{R}^{3}$}},
  journal   = {Journal of Optimization Theory and Applications},
  year      = {2024},
  volume    = {202},
  number    = {3},
  pages     = {1077--1100}
  }

@inbook{Amari-2009,
author = {Amari, Shun-Ichi},
title = {Information geometry and its applications: convex function and dually flat manifold},
year = {2009},
publisher = {Springer-Verlag},
address = {Berlin, Heidelberg},
pages = {75–102},
numpages = {28}
}

@book{Vaart-1998,
place={Cambridge},
series={Cambridge Series in Statistical and Probabilistic Mathematics}, title={Asymptotic statistics}, publisher={Cambridge University Press}, author={Vaart, Aad van der},
year={1998}
}

@article{Hel25,
  title={Extended generalized Marshall--Olkin model for dependent censoring},
  author={Helali, Salima},
  journal={Scandinavian Journal of Statistics},
  year={2025},
  publisher={Wiley Online Library}
}

@article{Cso89,
  title={Testing for exponential and Marshall-Olkin distributions},
  author={Csorgo, Sándor and Welsh, AH},
  journal={Journal of Statistical Planning and Inference},
  volume={23},
  number={3},
  pages={287--300},
  year={1989}
}

@article{Fei15,
  title={Analysis of dependent competing risks in the presence of progressive hybrid censoring using Marshall--Olkin bivariate Weibull distribution},
  author={Feizjavadian, Seyed Hamed and Hashemi, Reza},
  journal={Computational Statistics \& Data Analysis},
  volume={82},
  pages={19--34},
  year={2015},
  publisher={Elsevier}
}

@article{Bar24,
  title={Competing risks analysis for dependent causes using Marshall-Olkin bivariate generalized lifetime family},
  author={Barnwal, Vikas and Panwar, MS},
  journal={Communications in Statistics-Theory and Methods},
  volume={53},
  number={4},
  pages={1212--1240},
  year={2024},
  publisher={Taylor \& Francis}
}

@article{Mei07,
  title={Test of fit for Marshall--Olkin distributions with applications},
  author={Meintanis, Simos G},
  journal={Journal of Statistical Planning and inference},
  volume={137},
  number={12},
  pages={3954--3963},
  year={2007},
  publisher={Elsevier}
}

@article{Ingrid-2024,
author = {Van Keilegom, Ingrid and Kekeç, Elif},
title = {Estimation of the density for censored and contaminated data},
journal = {Stat},
volume = {13},
number = {1},
pages = {e651},
year = {2024}
}
\end{document}